\documentclass[journal]{IEEEtran}
\usepackage{amsthm, amssymb, amsmath}
\usepackage{verbatim, float}
\usepackage{mathrsfs}
\usepackage{graphics}
\usepackage[usenames,dvipsnames]{pstricks}
\usepackage{epsfig}
\usepackage{pst-grad} 
\usepackage{pst-plot} 
\usepackage{cases}
\usepackage{algorithm,algorithmic, blkarray}
\usepackage[noadjust]{cite}

\usepackage{url}
\usepackage[lofdepth, lotdepth]{subfig}

\theoremstyle{plain}
\newtheorem{theorem}{Theorem}
\newtheorem{lemma}{Lemma}

\newtheorem{corollary}{Corollary}
\newtheorem{proposition}{Proposition}
\newtheorem{definition}{Definition}

\theoremstyle{definition}
\newtheorem{example}{Example}
\newtheorem{remark}{Remark}

\newcommand{\hoang}[1]{\textcolor{black}{{#1}}}

\newcommand{\B}{{\mathcal B}}

\newcommand{\G}{{\mathcal G}}

\newcommand{\I}{{\mathcal I}}
\newcommand{\J}{{\mathcal J}}

\newcommand{\M}{{\mathcal M}}

\renewcommand{\S}{{\mathcal S}}
\newcommand{\T}{{\mathcal T}}
\renewcommand{\P}{{\mathcal P}}
\renewcommand{\O}{{\mathcal O}}

\newcommand{\bA}{{\boldsymbol A}}
\newcommand{\bB}{{\boldsymbol B}}

\newcommand{\bX}{\boldsymbol{X}}

\newcommand{\by}{{\boldsymbol y}}

\newcommand{\bw}{{\boldsymbol{w}}}

\newcommand{\bx}{{\boldsymbol{x}}}
\newcommand{\bz}{{\boldsymbol{z}}}

\newcommand{\bbZ}{{\mathbb Z}}

\newcommand{\define}{\stackrel{\mbox{\tiny $\triangle$}}{=}}

\newcommand{\et}{{\emph{et al.}}}

\newcommand{\NT}{$(N,L,F)$-TAS}
\newcommand{\NMOT}{$(N-1,L,F)$-TAS}
\newcommand{\NmT}{$(N_{\mathsf{max}},L,F)$-TAS}
\newcommand{\NPOT}{$(N+1,L,F)$-TAS}

\newcommand{\NnT}{$(N_0,L,F)$-TAS}
\newcommand{\NnET}{$(N_0,L,F)$-ETAS}
\newcommand{\NPT}{$(N',L,F)$-TAS}

\newcommand{\Nm}{N_{\sf{max}}}
\newcommand{\Nmi}{N_{\sf{min}}}
\newcommand{\Nmm}{[\Nmi, \Nm]}
\newcommand{\SNz}{\S^{N_0}}
\newcommand{\SN}{\S^{N}}
\newcommand{\SNm}{\S^{N_{\mathsf{max}}}}
\newcommand{\SNP}{\S^{N'}}
\newcommand{\SNMO}{\S^{N-1}}
\newcommand{\SNPO}{\S^{N+1}}
\newcommand{\SCNPO}{\S_{\sf{cyc}}^{N+1}}
\newcommand{\SCNMO}{\S_{\sf{cyc}}^{N-1}}
\newcommand{\SCN}{\S_{\sf{cyc}}^N}
\newcommand{\SDCNPO}{\S_{\delta\text{-}\sf{cyc}}^{N+1}}
\newcommand{\SDCNMO}{\S_{\delta\text{-}\sf{cyc}}^{N-1}}
\newcommand{\SDCN}{\S_{\delta\text{-}\sf{cyc}}^N}
\newcommand{\SDPCNPO}{\S_{\delta'\text{-}\sf{cyc}}^{N+1}}

\newcommand{\SDPCN}{\S_{\delta'\text{-}\sf{cyc}}^N}

\newcommand{\DNNMO}{\Delta_{N,N-1}}
\newcommand{\DNNPO}{\Delta_{N,N+1}}
\newcommand{\DNNP}{\Delta_{N,N'}}

\newcommand{\SNn}{S^N_n}
\newcommand{\SNns}{S^N_{n^*}}
\newcommand{\SNPOn}{S^{N+1}_n}
\newcommand{\SNMOn}{S^{N-1}_n}
\newcommand{\SNPn}{\S^{N'}_n}
\newcommand{\SNmn}{S^{N_{\mathsf{max}}}_n}
\newcommand{\SNmnp}{S^{N_{\mathsf{max}}}_{n'}}

\newcommand{\Gns}{\G_{n^*}}
\newcommand{\Uns}{U_{n^*}}
\newcommand{\Vns}{V_{n^*}}

\newcommand{\lfr}{$(L,F)$-ZWR}

\newcommand{\aaw}{\texttt{avg\_avg\_wasted}}
\newcommand{\amw}{\texttt{avg\_max\_wasted}}
\newcommand{\mmw}{\texttt{max\_max\_wasted}}
\newcommand{\aawp}{\texttt{aaw\_percentage}}
\newcommand{\amwp}{\texttt{amw\_percentage}}
\newcommand{\mmwp}{\texttt{mmw\_percentage}}
\newcommand{\fr}{\texttt{fraction}}

\begin{document}
\title{Transition Waste Optimization\\ for Coded Elastic Computing 
\thanks{
Son Hoang Dau and Zahir Tari are with the School of Computing Technologies, RMIT University. Emails: \{sonhoang.dau, zahir.tari\}@rmit.edu.au.
Ryan Gabrys is with the University of California, San Diego.
Email: gabrys@ucsd.edu. 
Yu-Chih Huang is with the Institute of Communications Engineering, National Yang Ming Chiao Tung University. Email: jerryhuang@nycu.edu.tw).
Chen Feng is with the School of Engineering, British Columbia University (Okanagan Campus). Email: chen.feng@ubc.ca. 
Quang-Hung Luu is with the Department of Computing Technologies, School of Science, Computing and Engineering Technologies, Swinburne University of Technology, and also with the Department of Civil Engineering, Monash University. Emails: hluu@swin.edu.au/hung.luu@monash.edu.
Eidah Alzahrani is with Albaha University. Email: ejalzahrani1@bu.edu.sa.
Part of this work was presented at the IEEE International Symposium on Information Theory, 2020.
}
\author{Son Hoang Dau, \emph{Member}, \emph{IEEE}, Ryan Gabrys, \emph{Member}, \emph{IEEE}, Yu-Chih Huang, \emph{Member}, \emph{IEEE}, Chen Feng, \emph{Member}, \emph{IEEE}, Quang-Hung Luu, \emph{Member}, \emph{IEEE}, Eidah Alzahrani, Zahir Tari}
}\vspace{-20pt}
\maketitle 

\begin{abstract}
Distributed computing, in which a resource-intensive task is divided into subtasks and distributed among different machines, plays a key role in solving large-scale problems. 
\emph{Coded computing} is a recently emerging paradigm where redundancy for distributed computing is introduced to alleviate the impact of slow machines (stragglers) on the completion time. 
We investigate coded computing solutions over elastic resources, where the set of available machines may change in the middle of the computation. 
This is motivated by recently available services in the cloud computing industry (e.g., EC2 Spot, Azure Batch) where low-priority virtual machines are offered at a fraction of the price of the on-demand instances but can be preempted on short notice.
Our contributions are three-fold. We first introduce a new concept called \emph{transition waste} that quantifies the number of tasks existing machines must abandon or take over when a machine joins/leaves. We then develop an efficient method to minimize the transition waste for the cyclic task allocation scheme recently proposed in the literature (Yang \emph{et al.} ISIT'19).
Finally, we establish a novel solution based on finite geometry achieving \emph{zero} transition wastes given that the number of active machines varies within a fixed range.  
\end{abstract}
\vspace{-5pt}

\section{Introduction} 
\label{sec:intro}

In the era of Big Data, massive computational tasks, e.g. in large-scale machine learning and data analytics, are often carried out in \emph{distributed systems} like Apache Spark~\cite{Zaharia_etal_2010} and Hadoop~\cite{Hadoop}, which can efficiently process terabytes or even petabytes of data. However, it has been observed in such systems that slow machines, or \emph{stragglers}, which may run 6x-8x slower than a median one, may significantly affect the performance of the whole distributed system~\cite{Ananthanarayanan_etal_2013, Dean_etal_2013, Yadwadkar_etal_2016}. 

\emph{Coded distributed computing}~\cite{Lee_etal_2018, Li_etal_2015, Tandon_etal_2017}, built upon algorithmic fault tolerance~\cite{Huang_etal_1984}, is a recently emerging paradigm where computation redundancy 
is employed to tackle the straggler effect. As a toy example~\cite{Lee_etal_2018}, to perform a matrix-vector multiplication $\bA \bx$, a master machine first partitions the matrix $\bA$ into two equal-sized submatrices $\bA_1$ and $\bA_2$ and then distributes $\bA_1$, $\bA_2$, and $\bA_1+\bA_2$ to three worker machines, respectively. These machines also receive the vector $\bx$ and perform three multiplications $\bA_1\bx$, $\bA_2\bx$, and $(\bA_1+\bA_2)\bx$ in parallel. Clearly, $\bA\bx$ can be recovered by the master from the outcomes of any two workers. Thus, this coded scheme can tolerate one straggler. 
The potential of coded distributed computing has been extensively investigated through a substantial body of work in the literature, e.g.,~\cite{Dutta_etal_2016, Yu_etal_2017, Karakus_etal_2017, Li_etal_2018, Mallick_etal_2018}. Recent breakthroughs have shown that this paradigm works for not only linear/bilinear operations but also general nonlinear operations such as polynomial evaluation~\cite{Yu_etal_2019} or even any function that can be represented by a deep network~\cite{Kosaian_etal_2018, Kosaian_etal_2020}.

Most of the research in the literature of coded distributed computing, however, assumes that the set of available worker machines remains \emph{fixed}. This critical limitation renders current coded computing schemes \emph{inapplicable} in an environment where low-cost elastic resources are readily available.
In fact, major cloud computing providers, very recently, started 
offering spare virtual machines at a price up to 90\% cheaper than that of the on-demand machines, e.g. \emph{Amazon EC2 Spot}~\cite{AmazonSpot} and \emph{Microsoft Azure Batch}~\cite{AzureBatch}, albeit at the cost of low priority in the sense that these machines can be preempted (removed) for a higher-priority customer under a short notice (e.g., two minutes in the case of Amazon Spot). This new development in the cloud computing industry provides customers with an opportunity to have large computing resources at a fraction of the cost of the standard on-demand service. Realizing this opportunity, however, requires the user to develop much more flexible distributed computing paradigms in order to efficiently exploit \emph{elastic resources} where low-cost machines can leave and join at any time during the computation cycle. 

Recently, Yang \emph{et al.}~\cite{Yang_etal_2019} proposed an elegant technique extending coded computing to deal with elastic resources.
Their key idea is to \emph{couple} a cyclic task allocation scheme, which works for any  number of machines, 
with a coded computing scheme to guarantee that a) as long as there are a sufficient number of machines working, the scheme can tolerate stragglers, and b) the workload at each machine is inversely proportional to the number of available machines. In other words, their solution allows an elastic task allocation: when a new machine joins, existing machines can share some of their workload with the newcomer, reducing the number of tasks allocated to them; likewise, when a machine leaves, existing machines must cover extra tasks left over. The elastic coded computing scheme proposed in~\cite{Yang_etal_2019} was evaluated in the multi-tenancy cluster at Microsoft using the Apache REEF Elastic Group Communication framework, and shown to reduce the completion time of matrix-vector multiplication and linear regression by up to 46\% compared to ordinary coded computing schemes. 
The elastic setting was later extended to cover machines with heterogeneous computation speeds and storage capacities in~\cite{WoolseyChenJi-ISIT-2020, WoolseyChenJi-TCOM-2021, WoolseyKliewerChenJi-Globecom-2021}.
A combination with hierarchical coding was also proposed in~\cite{KianiAdikariDraper-ICASSP-2021}, which allows finer task allocations to speed up the completion time. 
 
Relaxing the cyclic task allocation proposed in~\cite{Yang_etal_2019}, we investigate a more general \emph{elastic task allocation} problem, which we believe may find applications not just in coded distributed computing but also in a much broader context where a set of tasks is distributed to an elastic set of participants (e.g., virtual machines), which frequently leave and join. More specifically, we seek to address the following key questions. 
\begin{itemize}
	\item \emph{Task allocation}: given a set of tasks and a set of machines, how do we assign tasks to machines so that all machines are assigned an equal number of tasks (workload balance) and every task is covered by the same number of machines? \hoang{These combinatorial constraints can be easily satisfied}, e.g., by using the cyclic scheme employed in~\cite{Yang_etal_2019} or its generalization. 
	\item \emph{Transition reallocation}: when an elastic event occurs (machines leaving/joining), how do we reallocate the tasks to the new set of machines to minimize the \emph{transition waste}, i.e., the total number of tasks that existing machines have to abandon or take over when one machine joins or leaves, \hoang{minus} the necessary amount? \hoang{As an optimization problem, this is a much more challenging question compared to task allocation and is the focus in our work.}    
\end{itemize}


\begin{figure}[htb!]
	\centering
	\begin{center}
    \subfloat[Cyclic task allocation for five machines~\cite{Yang_etal_2019}. \label{subfig-1:dummy}]{
    \tabcolsep=0.25cm
	\begin{tabular}{|c|c|c|c|c|}
	\hline      
	$S^5_1$ & $S^5_2$ & $S^5_3$ & $S^5_4$ & $S^5_5$\\
    \hline
    $0\to 3$ &  &  & $0\to 3$ & $0\to 3$\\
    \hline
    $4\to 7$ & $4\to 7$ &  &  & $4\to 7$ \\
    \hline
    $8\to 11$ & $8\to 11$ & $8\to 11$ &  & \\
    \hline
     & $12\to 15$ & $12\to 15$ & $12\to 15$ & \\
    \hline
     &  & $16\to 19$ & $16\to 19$ & $16\to 19$\\
    \hline
     \end{tabular}
     }
     \end{center}
     \subfloat[Cyclic task allocation for four machines~\cite{Yang_etal_2019}. The \emph{transition waste} from five to four machines is \emph{12 tasks}. \label{subfig-2:dummy}]{
    \tabcolsep=0.45cm
    \begin{tabular}{|c|c|c|c|}
	\hline      
	$S^4_1$ & $S^4_2$ & $S^4_3$ & $S^4_4$\\
    \hline
    $0\to 4$ &  & $0\to 4$ & $0\to 4$\\
    \hline
    $5\to 9$ & $5\to 9$ &  & $5\to 9$ \\
    \hline
    $10\to 14$ & $10\to 14$ & $10\to 14$ & \\
    \hline
     & $15\to 19$ & $15\to 19$ & $15\to 19$\\
    \hline
     \end{tabular}
     }
     \hspace{10pt}
     \subfloat[Our proposed \emph{shifted} cyclic task allocation for four machines that results in an \emph{optimal transition waste} among all cyclic schemes (\emph{zero} in this case). \label{subfig-2:dummy}]{
     \tabcolsep=0.45cm
    \begin{tabular}{|c|c|c|c|}
	\hline      
	$S^4_1$ & $S^4_2$ & $S^4_3$ & $S^4_4$\\
    \hline
    $17\to 1$ &  & $17\to 1$ & $17\to 1$\\
    \hline
    $2\to 6$ & $2\to 6$ &  & $2\to 6$ \\
    \hline
    $7\to 11$ & $7\to 11$ & $7\to 11$ & \\
    \hline
     & $12\to 16$ & $12\to 16$ & $12\to 16$\\
    \hline
     \end{tabular}
     }
     \caption{Illustration of the sub-optimality of the cyclic task allocation scheme proposed in~\cite{Yang_etal_2019} with respect to the \emph{transition waste} when Machine 5 leaves. Here, we use $a\to b$ to denote the set $\{a,a+1,\ldots,b\}\pmod{F}$, where $F=20$.} 
     \label{fig:toy_example}
\end{figure}

We illustrate in a toy example (Fig.~\ref{fig:toy_example}) the  \emph{transition waste} concept and explain why the cyclic elastic task allocation scheme in~\cite{Yang_etal_2019} is \emph{suboptimal} with respect to this new metric. 
We consider the computation of $\bA\bx$ where $\bA$ can be partitioned into 40 equal-sized sub-matrices $\bA_1,\ldots,\bA_{40}$. We first group these sub-matrices into 20 groups, e.g., $\{\bA_1,\bA_2\}$, $\{\bA_3,\bA_4\}$, and so forth. Then each group is assigned a \emph{task index} 
from $0$ to $19$.
Task 0, for instance, corresponds to the computation of $\{\bA_1\bx,\bA_2\bx\}$. Task 0 is \emph{encoded} into five subtasks: $\bA_1 \bx$, $\bA_2\bx$, $(\bA_1 + \bA_2)\bx$, $(\bA_1 + 2\bA_2)\bx$, and $(\bA_1 + 3\bA_2)\bx$. \emph{A machine taking Task 0 means it computes one of these five subtasks}. Similar to the earlier discussion, any \emph{three} out of five subtasks/machines form a \emph{coded computing group} that can recover Task 0 given one straggler. 

Hence, abstracting away the underlying coded computing scheme, which can be designed \emph{independently} of the task allocation scheme in consideration, given $F=20$ tasks, we require that \emph{each task must be covered by precisely $L=3$ machines}. This requirement can be met by using the cyclic scheme in~\cite{Yang_etal_2019}: each of the $N$ machines is preloaded with a set of $F$ tasks, which is then divided into $N$ equal consecutive subsets of size $F/N$ each, and each machine works on tasks in the union of $L$ consecutive such subsets.
For instance, when $N=5$, Machine 1 works on the set of tasks $S_1^5 = \{0,1,\ldots,11\}=\{0,\ldots,3\}\cup\{4,\ldots,7\}\cup \{8,\ldots,11\}$, Machine 2 works on $S_2^5 = \{4,5,\ldots,15\}$, and so on (Fig.~\ref{fig:toy_example}~(a)). Note that each machine takes 12 tasks and due to the cyclic task allocation scheme, each task is covered by three machines. 

In Fig.~\ref{fig:toy_example}~(b), only four machines are available, each of which takes 15 tasks. As Machine 5 has left, it is necessary now that each of the four available machines must take $3=15-12$ more tasks. Ideally, when the transition from five machines to four machines occurs, each machine continues their existing tasks and works on three new tasks. This is true for Machine~1 because $S^5_1 \subset S^4_1$. However, it is not the case for other machines. For instance, Machine~3 has to abandon \emph{two} tasks (8 and 9) and takes over \emph{five} new tasks $(0\to 4)$. The transition waste at Machine~3 is $(2+5)-3 = 4$ tasks. Note that three is the \emph{necessary} increase in the number of tasks each machine must take and so we subtract that amount. The transition wastes at other machines can be computed in a similar manner. The total transition waste is
\[
(0+3-3) + (1+4-3) + (2+5-3) + (3+6-3) =12 \text{ (tasks)},
\] 
Therefore, sticking to the cyclic allocation scheme of~\cite{Yang_etal_2019}, we waste 12 tasks. However, it turns out that the transition waste can be reduced to \emph{zero} if we use the allocation scheme in Fig.~\ref{fig:toy_example}~(c) instead. In this case, as $S^5_n \subset S^4_n$, the transition wastes at all four machines are zero. The trick is to \emph{shift} the cyclic task allocation by a right amount ($-3$ in this case) to maximize the overlaps between $S^5_n$ and $S^4_n$, $n = 1,\ldots, 4$.

Our main contributions are summarized below.
\begin{itemize}
	\item We first introduce the new concept of \emph{transition waste} of an elastic task allocation scheme, which quantifies the total number of tasks that existing machines must abandon or take over when one machine joins or leaves, less the necessary amount. \emph{A reduction in transition waste implies lower computation and communication costs} (Remark~\ref{rm:TW}).  
	\item We then compute explicitly the transition waste incurred in the cyclic elastic task allocation scheme introduced by Yang \emph{et al.}~\cite{Yang_etal_2019} when machines leave and join (Theorems~\ref{thm:joining},~\ref{thm:leaving}) and propose \emph{shifted} cyclic schemes that minimize the transition waste among all cyclic schemes (Theorems~\ref{thm:shifted_joining},~\ref{thm:shifted_leaving}). The optimal transition waste of a shifted cyclic scheme is, in general, greater than zero. 
	\item Lastly, we show that there exists a \emph{zero-waste} transition when a machine leaves if and only if there exists a \emph{perfect matching} in a certain bipartite graph, using the famous Hall's marriage theorem. Based on this new insight, we construct several novel task allocation schemes based on \emph{finite geometry} that achieve \emph{zero transition wastes} when the number of active machines varies within a fixed range. 
\end{itemize}
While the cyclic schemes are simple to implement and efficient when there are many tasks and many machines, the schemes with zero-waste transitions are more suitable when there are a moderate number of machines and tasks but each task is resource-intensive. We will discuss this further in Sections~\ref{sec:pre}. 

We emphasize that our task allocation schemes are designed \emph{separately} from the underlying coded computing scheme and hence can be applied on top of existing coded computing schemes (see Section~\ref{subsec:coupling}). The readers who are familiar with the \emph{parity declustering} technique in redundant disk arrays (RAID)~\cite{MuntzLui1990, HollandGibson1992, DauJiaJinXiChan2014} may recognize the analogy between a coded computing scheme and a stripe unit and between a task allocation scheme and a data layout~(in the terminology of~\cite{HollandGibson1992}).

The paper is organized as follows. 
The concepts of elastic task allocation and transition waste are defined and discussed in Section~\ref{sec:pre}. Section~\ref{sec:cyclic} is devoted for the cyclic task allocation scheme and our proposed shifted version with optimal transition wastes. We develop elastic task allocation schemes that admit zero transition wastes in Section~\ref{sec:zero}, perform simulations and experiments in Section~\ref{sec:evaluations}, and conclude the paper in Section~\ref{sec:conclusions}. 

\section{Preliminaries}
\label{sec:pre}

\subsection{Elastic Task Allocation Scheme}
\label{subsec:ETAS}

We define in this section the \emph{elastic task allocation scheme}, which generalizes the cyclic scheme originally proposed by Yang \emph{et al.}~\cite{Yang_etal_2019}, and the new concept of the \emph{transition waste}. 
\hoang{Frequently used notations can also be found in Appendix~\ref{app:notations}.}

We henceforth use $N$ for the number of available machines, $F$ for the common number of pre-loaded tasks at each machine, and $L$ as minimum number of available machines so that the scheme still works ($L \leq N$). Each task is represented by a label from $[[F]] \define \{0,1,\ldots,F-1\}$. We assume that all tasks consume an equal amount of resources (storage, memory, CPU). We use $[F]$ to denote the set $\{1,2,\ldots,F\}$ and $[A,B]$ to denote the set $\{A,A+1,\ldots,B\}$. We also use $2^{[[F]]}$ to denote the power set
of the set $\{0,1,\ldots,F-1\}$ and $(2^{[[F]]})^ N= 2^{[[F]]} \times 2^{[[F]]} \times \cdots \times 2^{[[F]]}$ to denote the 
$N$-ary Cartesian power of $2^{[[F]]}$. 

\begin{definition}[Task allocation scheme] 
\label{def:TAS}
An ordered list of $N$ sets $\S^N = (S^N_1,\ldots,S^N_N)\in (2^{[[F]]})^N$, where $\SNn \subset [[F]]$, $n \in [N]$, is referred to as an $(N,L,F)$ \emph{task allocation scheme} (\NT) if it satisfies the following two properties. 
\begin{itemize}
	\item ($L$-Redundancy) each element in $[[F]]$ is included in precisely $L$ sets in $\S^N$, and 
	\item (Load Balancing) $|\SNn| = LF/N$ for all $n \in [N]$. Here we assume that $LF/N \in \bbZ$. 
\end{itemize}
\end{definition} 

Note that we can relax the Load Balancing property and require that $\SNn \in \{\lfloor LF/N \rfloor, \lceil LF/N \rceil\}$ and hence can lift the requirement that $N$ divides $LF$. To simplify the exposition, however, we assume $LF/N \in \bbZ$. In practice, padding of dummy tasks can be employed to achieve this property.
The $L$-Redundancy property is tied to the underlying coded computing scheme (see Section~\ref{subsec:coupling}). 

An \NT~$\S^N = (S^N_1,\ldots,S^N_N)$ can also be represented by its \emph{incidence matrix} $\bB = (b_{f,n})_{F\times N}$, where $b_{f,n} = 1$ if and only if $f \in \SNn$. The rows and columns of $\bB$ represent tasks and machines, respectively. Clearly, $\bB$ has row weight $L$ and column weight $LF/N$. In other words, each row of $\bB$ has precisely $L$ ones while each column has precisely $LF/N$ ones. Thus, a TAS simply corresponds to a binary matrix with constant row and column weights. 

\begin{example} 
\label{ex:toy2}
For $N = 3, L = 2, F = 6$, the list of sets 
\[
\S^3 = (\{0,1,2,3\}, \{2,3,4,5\}, \{4,5,0,1\}) 
\]
is a $(3,2,6)$-TAS as each member set has size $4 = 2\times 6/3$ and each element $f \in [[6]] = \{0,1,\ldots,5\}$ belongs to precisely $L = 2$ such sets. 
The incident matrix of $\S^3$, given by \eqref{mat:S3}, has column weight four and row weight two.
\begin{equation}
\label{mat:S3}
\begin{blockarray}{cccc}
& \text{Machine 1} & \text{Machine 2} & \text{Machine 3}\\
& S^3_1 & S^3_2 & S^3_3\\
\begin{block}{c(ccc)}
\text{Task 0} & 1 & 0 & 1\\
\text{Task 1} & 1 & 0 & 1\\
\text{Task 2} & 1 & 1 & 0\\
\text{Task 3} & 1 & 1 & 0\\
\text{Task 4} & 0 & 1 & 1\\
\text{Task 5} & 0 & 1 & 1\\ 
\end{block}
\end{blockarray}\ .
\end{equation}
\end{example} 
\vspace{-5pt}

When a machine leaves or joins, we need to reallocate tasks to a new set of machines. Thus, we must extend the notion of a task allocation scheme (TAS) to that of an \emph{elastic} task allocation scheme (ETAS). 
We explain in Section~\ref{subsec:coupling} how to couple an ETAS and a coded computing scheme to achieve a coded elastic computing scheme that tolerates stragglers. 
\hoang{Note that the parameter $L$ in Definition~\ref{def:ETAS} is specified by the underlying coded computing scheme and considered fixed in an elastic task allocation scheme while the number of machines $N$ and the number of tasks $F$ can vary.}

\begin{definition}[Elastic task allocation] 
\label{def:ETAS}
A pair $(\SNz,\T)$ is referred to as an $(N_0,L,F)$ elastic task allocation scheme (\NnET) if $\SNz$ is the initial \NnT~and $\T$ is an algorithm that reallocates tasks when machines leave and join so that the new scheme remains a TAS. More specifically,
\[
\T \colon (2^{[[F]]})^N \times \{-1,1\} \times [N] \to (2^{[[F]]})^{N-1}\cup (2^{[[F]]})^{N+1}
\]
takes as input an \NT~$\SN$, where $L \leq N \leq LF$, a variable $b \in \{-1,1\}$, which represents the elastic event of one machine leaving ($b=-1$) or joining ($b=1$), and an index $n^* \in [N]$, which indicates the index of the machine that leaves when $b = -1$ (when $b = 1$, $n^*$ is ignored). Moreover, $\T$ returns an output $\SNP$, which is another \NPT, where $N' = N+b$. In other words, moving from a set of $N$ machines to a new set of $N'=N+b$ machines, $\T$ updates the list of task sets $\SN$ to obtain $\SNP$, which remains a TAS. 
\end{definition} 

A few remarks are in order. 
\emph{First}, we make a simplifying assumption in Definition~\ref{def:ETAS} that each elastic event corresponds to \emph{one} machine leaving and joining only. In other words, we assume that machines leave and join one after another and not at the same time.
\hoang{This not only allows us to avoid complex mathematical notations  but also covers the case of multiple machines leaving/joining: we can treat that case as a series of independent transitions in each of which only one machine leaves or joins.}
\emph{Second}, while in general we allow $N$ to take any value in the range $[L,LF]$, it is more practical to limit $N$ within a fixed range $[L,\Nm]$. Moreover, we often assume that $F$ is divisible by any number within this range. These assumptions allow us to achieve concrete results and are also practically reasonable. For instance, we can use padding, i.e., adding dummy tasks, to make $F$ satisfy the aforementioned property. 
\emph{Third}, when Machine $n^*\in[N]$ leaves, we index the remaining machines by the set $[N-1] = \{1,\ldots,N-1\}$. However, when comparing with the previous TAS, we often use $\{1,\ldots,n^*-1,n^*+1,\ldots,N\}$, instead of $[N-1]$, so that the same machine is given the same index in the previous and in the current task allocation schemes.  

\textbf{Cyclic elastic task allocation scheme~\cite{Yang_etal_2019}.} A simple way to construct an ETAS is to let $\T$ depend only on the number of machines and not on the current TAS. More specifically, whenever there are $N$ machines available as the result of an elastic event, we always use a fixed \NT
\begin{equation}
\label{eq:cyclic}
\begin{split}
\SCN &= (S^N_1,\ldots,S^N_N),\\
\SNn &= \left[(n-1)\frac{F}{N}, (n-1)\frac{F}{N} + \frac{LF}{N}-1\right] \pmod F,
\end{split}
\end{equation}
for every $n \in [N]$, where $[A,B] \pmod F$ is obtained from $[A,B]$ by applying the modulo operation on every element of this set. We also assume here that $F/N \in \bbZ$. 

It is straightforward to verify that each $\SCN$ satisfies the Load Balancing and the $L$-Redundancy properties, and therefore, is indeed an $(N,L,F)$-TAS. The reallocation algorithm is trivial: $\T(\SCN,1) = \SCNPO$ and $\T(\SCN,-1,n^*) = \SCNMO$ for every $n^* \in [N]$. Fig.~\ref{fig:toy_example}(a) and (b) illustrate the cyclic ETAS when $N=5$ and when $N=4$, $L = 3$, $F=20$, and $n^*=5$. 

\subsection{Transition Waste} 
\label{subsec:tw}

We now define the \emph{transition waste} during an elastic event when one machine leaves or joins and demonstrate this new concept via a few examples. 

\begin{definition}[Necessary load change]
\label{def:nlc} 
For a transition from an \NT~$\SN$ to another \NPT~$\SNP$, $\DNNP \define |LF/N-LF/N'|$ is referred to as the necessary load change. 
When $N' = N\pm 1$, we have $\Delta_{N,N\pm1} = LF/(N(N\pm1))$.
\end{definition} 

The \emph{necessary load change}, $\DNNP = ||\SNn|-|\SNPn||$, reflects the necessary increase or decrease in the number of tasks each machine must take when one machine leaves or joins, respectively. For instance, when $L = 3, F = 20$, if there are $N = 5$ machines, the Load-Balancing property requires that each machine runs $LF/N = 12$ tasks, while if there are $N' = 4$ machines due to the removal of one, then each machine runs $LF/N' = 15$ tasks. Therefore, each of the four machines has to take $3 = 15-12$ more tasks to react to this event. The necessary load change is \emph{three} in this case. 

\begin{definition}[Transition waste for one machine] 
\label{def:TW_one}
The transition waste incurred at Machine $n$ when transitioning from a set of tasks $\SNn$ to another set of tasks $\SNPn$ is defined as
\[
W(\SNn \to \SNPn) = |\SNn \Delta \SNPn|-\DNNP,
\]
where $\DNNP$ is the necessary load change (Definition~\ref{def:nlc}) and $A\Delta B$ denotes the symmetric difference between $A$ and $B$. We also use $W_{n^*}(\SNn \to \SNPn)$ for the case Machine $n^*$ leaves. 
\end{definition}  

\begin{remark}
In Definition~\ref{def:TW_one}, we assume that each existing machine keeps its current index in the new TAS, that is, $\SNn$ and $\SNPn$ refer to the task sets assigned to the same machine.
\end{remark}

\begin{remark}
\label{rm:TW}
Note that
$|\SNn \Delta S^{N'}_n| = |\SNn \setminus S^{N'}_n| + |S^{N'}_n \setminus \SNn|$
corresponds to the number of scheduled tasks Machine $n$ has to abandon (tasks in $\SNn$ but not in $S^{N'}_n$) and take anew (tasks in $S^{N'}_n$ but not in $\SNn$). Thus, the transition waste $W(\SNn \to \SNPn)$ in Definition~\ref{def:TW_one} measures the maximum number of tasks wasted at Machine $n$ when another machines leaves or joins. As some tasks may have been already completed before the transition, one should abandon as few existing tasks as possible. Likewise, taking on fewer new tasks will decrease the downloading traffic (if the protocol requires new tasks to be downloaded). Thus, a \emph{low-waste transition} saves computation and network resources and hence \emph{reduces the completion time}. 
\end{remark}

The transition waste of a TAS is defined as the total transition wastes at all machines.

\begin{definition}[Transition waste]
\label{def:TW}
When Machine $N+1$ joins, the transition waste of the transition from an \NT~$\SN$ to an \NPOT~$\SNPO$ is defined as
\[
W(\SN\to \SNPO) \define \sum_{n\in [N]} W(\SNn\to \SNPOn).
\]
When Machine $n^*$ leaves, the transition waste of the transition from an \NT~$\SN$ to an \NMOT~$\SNMO$ is defined as\vspace{-5pt}
\[
W_{n^*}(\SN\to \SNMO) \define \sum_{n \in [N]\setminus \{n^*\}} W_{n^*}(\SNn\to \SNMOn).
\]
Here, $W(\SNn\to \SNPOn)$ and $W_{n^*}(\SNn\to \SNMOn)$ denote the transition waste incurred at Machine $n$ (Definition~\ref{def:TW_one}). 
\end{definition} 

We demonstrated in Fig.~\ref{fig:toy_example}(a), (b), (c)) in Section~\ref{sec:intro} two different transitions from a $(5, 3, 20)$-TAS to a $(4, 3, 20)$-TAS, i.e., one machine removed. The first transition has a transition waste of 12 tasks, while the second one has a zero waste. Another example, built upon Example~\ref{ex:toy2}, is given below. 

\begin{example} 
\label{ex:toy3}
Let $L = 2,F = 6,N = 3,N' = 4$. We can verify that  
$
\S^3 = (\{0,1,2,3\}, \{2,3,4,5\}, \{4,5,0,1\})
$
is a $(3,2,6)$-TAS and 
$
\S^4 = (\{0,1,2\}, \{0,1,2\}, \{3,4,5\}, \{3,4,5\})
$ 
is a $(4,2,6)$-TAS. The necessary load change when going from three to four machines, and vice versa, is $\Delta_{3,4} = |4-3| = 1$. The waste when transitioning from $\S^3$ to $\S^4$ is computed as follows. \vspace{-5pt}
\[
\begin{split}
W(\S^3\to \S^4) &= \sum_{n=1}^3 (|S^3_n \Delta S^4_n|-\Delta_{3,4})\\
&= (1-1) + (5-1) + (3-1) = 6.
\end{split}
\] 
\end{example} 

\subsection{Coupling an Elastic Task Allocation Scheme and a Coded Computing Scheme}
\label{subsec:coupling}

We now explain how to couple an elastic task allocation scheme (ETAS) and a coded computing scheme (CCS) to achieve a coded elastic computing scheme, which allows 
\begin{itemize}
\item \emph{straggler tolerance}: at most $E$ slow machines do not affect the completion time of the system,
\item \emph{load balancing}: every available machine is assigned the same workload,
\item \emph{elasticity}: the workload of available machines can be flexibly adjusted when machines leave and join. 
\end{itemize}

The general method is to first partition the problem instance into $F$ \emph{independent} sub-instances and then apply a CCS to \emph{each} sub-instance. Each task~$f \in [[F]]$ refers to the computation task performed over the $f$th sub-instance. 
Suppose that throughout the computation the number of available machines varies from $L$ to $\Nm$. 
For each task, a CCS generates $\Nm$ sub-tasks, which are distributed to maximum $\Nm$ machines so that the completion of any $L-E$ sub-tasks leads to the completion of the task ($L-E$ is referred to as the \emph{recovery threshold}). 
Each of the $N$ available machines must be loaded with the corresponding sub-tasks of \emph{all} $F$ sub-instances so as to be ready to work on any new tasks when machines leave or join. However, each machine only works on the sub-tasks of the tasks assigned to it by the TAS.
More specifically, if an \NT~$\SN = \{S^N_1,\ldots,S^N_n\}$ is used then Machine $n$ only works on tasks indexed by $S^N_n$.

The $L$-Redundancy of the TAS guarantees that any task $f$ is worked on by precisely $L$ different machines among $N$. 
As the CCS allows the recovery of Task~$f$ from any $L-E$ outputs, the coded elastic computing scheme, which couples a TAS and a CCS, can tolerate $E$ stragglers. 
The Load Balancing property of the TAS guarantees that every available machine is assigned the same workload.
When a machine joins or leaves, a new TAS constructed by the transition algorithm $\T$ of the ETAS is applied, which preserves the straggler tolerance and the load balancing property. 
We discuss below how to define the tasks for the \textit{matrix-vector multiplication} problem. For other related problems such as matrix-matrix multiplication, linear regression, and multivariate polynomial evaluation please refer to Appendix~\ref{app:coupling}.   

\textbf{Matrix-Vector Multiplication.}
We aim to compute $\bA \bx$, where $\bA$ is a matrix and $\bx$ is a vector of matching dimension, in a way that tolerates any $E$ stragglers $(0 \leq E < L)$, and with a varied number of available machines $N$ $(L \leq N \leq \Nm)$.

Assuming that the number of rows of $\bA$ is divisible by $F$ (padding if necessary), we partition $\bA$ row-wise into $F$ equal-sized sub-matrices $\bA_0,\bA_1,\ldots,\bA_{F-1}$.
The pair $(\bA_f,\bx)$ forms the $f$th sub-instance of the original instance $(\bA,\bx)$ and the computation of $\bA_f\bx$ is referred to as Task~$f$.
A known CCS for matrix-vector multiplication (e.g.,~\cite{Lee_etal_2018}) can then be used to generate $\Nm$ sub-tasks for each Task~$f$, each of which is then distributed to the corresponding machine (machines joining later download later). 
Clearly, the completion of all tasks $f \in [[F]]$ gives us the desired product $\bA\bx$.

\subsection{Storage, Communication, and Computation Overhead of an ETAS} 
\label{subsec:costETAS}

As proposed in~\cite{Yang_etal_2019}, each machine stores all $F$ tasks but only runs a subset of those tasks based on the specific allocation. In this way, when switching to a new TAS, each existing machine doesn't have to download new data.
When coupling with a coded computing scheme (see Section~\ref{subsec:coupling}), each machine actually stores only a $1/(L-E)$-fraction of the input data, e.g., the matrix $\bA$ if we are computing $\bA\bx$, where $E < L$ is the number of stragglers the scheme can tolerate. 

Every machine joining the system has to download its portion of data once, which constitutes the most costly, but necessary, communication overhead of the system.
From a practical perspective, letting a machine \textit{joining} in the middle of the computation might be costly as it must download its allocated tasks from existing servers or from the master (decoding/re-encoding may even be needed). Machines leaving, on the other hand, would not cause any issue in terms of communication overhead because the active ones already stored all needed data and are ready to transition to any new sets of tasks. 
The communication between a master machine, which coordinates the task allocation, and the worker machines, is negligible. 

The master has to run an algorithm to find a new TAS whenever a machine leaves or joins. If a cyclic or a shifted cyclic ETAS (see Section~\ref{subsec:TW_shifted_cyclic}) is used, the computation overhead is negligible. If a zero-waste transition (see Section~\ref{sec:zero}) is insisted, the complexity of the search is polynomial in $N$, $L$, and $F$ (basically, it runs a network flow algorithm). \hoang{A zero-waste transition will be particularly beneficial when there are a moderate number of tasks while each task is resource-intensive. 
In that case, the benefit of a zero-waste transition will offset the time spent for finding one. Note that we have total control of the number of tasks $F$ when defining tasks (only requiring that $F$ satisfies some divisibility condition). For example, when computing the matrix-matrix multiplication $\bA \bB$  (see Appendix~\ref{app:coupling}), we can partition both $\bA$ and $\bB$ into a desirable number $F$ of sub-matrices $\bA_f$ and $\bB_f$ (row-wise for $\bA$ and column-wise for $\bB$). Then, Task $f$ corresponds to $\bA_f\bB_f$, $f=0,1,\ldots,F-1$.}

\section{Shifted Cyclic Elastic Task Allocation Schemes with Optimal Transition Wastes}
\label{sec:cyclic}

We first compute explicitly the transition waste of the cyclic elastic task allocation scheme introduced by Yang \emph{et al.}~\cite{Yang_etal_2019} (Theorems~\ref{thm:joining},~\ref{thm:leaving}) and then propose a shifted cyclic scheme that achieves the optimal transition waste among all such cyclic schemes (Theorems~\ref{thm:shifted_joining},~\ref{thm:shifted_leaving},~\ref{thm:optimal}). We assume that the number of machines $N$ lies in a predetermined interval $[L,\Nm]$ and $N(N+1)$ divides $F$ for every $L \leq N < \Nm$. 
\subsection{Transition Waste of the Cyclic Elastic Task Allocation}
\label{subsec:TW_cyclic}

The following lemma is useful in determining the symmetric difference between two sets in $[[F]]$. 

\begin{lemma} 
\label{lem:sym}
Let $S = [a,b]\pmod F$ and $T = [c,d]\pmod F$. Assume that $0 \leq a \leq c < F$, and moreover, $0 < |S| < F$ and $0 < |T| < F$. The following statements hold. 
\begin{itemize}
	\item[(a)] If $c-a < |S| < (c-a)+|T| < F$ then
	\[
	|S \Delta T| = 2(c-a) + |T|-|S|.
	\] 
	\item[(b)] If $|S| \geq (c-a) + |T|$ then $T \subset S$ and
	\[
	|S \Delta T| = |S|-|T|.
	\]
\end{itemize}
\end{lemma} 
\begin{proof} \textbf{(a)}
Suppose that $c-a <|S| < (c-a)+|T| < F$. If we travel along the circle of integers mod $F$ (see Fig.~\ref{fig:circle} (a)) clockwise from $a$, we first see $c$, then $b \pmod F$ (because $c-a < |S|$), then $d \pmod F$ (because $|S| < (c-a)+|T|$), before we reach $a$ again (because $(c-a)+|T| < F$). Therefore,    
\[
\begin{split}
|S \Delta T| &= |S \setminus T| + |T \setminus S|\\ 
&= (c-a) + (|T|-(|S|-(c-a)))\\
&= 2(c-a) + |T|-|S|.
\end{split}
\] 
\vspace{-5pt}
\begin{figure}[!htb]
\centering
\includegraphics[scale=0.6]{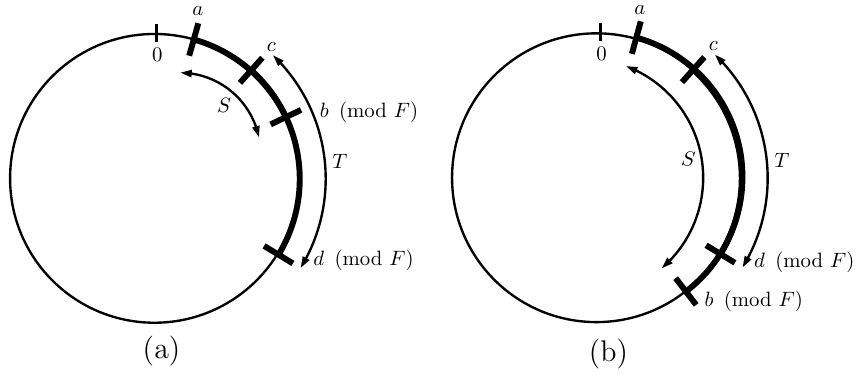}
\caption{Illustrations of the two sets $S = [a,b]\pmod F$ and $T = [c,d]\pmod F$ on the circle of integers mod $F$.}
\label{fig:circle}
\end{figure}

\textbf{(b)} Suppose that $|S| \geq (c-a)+|T|$. This clearly implies that $T \subset S$ and hence $|S\Delta T| = |S|-|T|$ (see Fig.~\ref{fig:circle} (b)). 
\end{proof} 

Lemma~\ref{lem:zero_trivial} identifies the case where the zero waste is achieved, which is obvious by the definition of the transition waste.

\begin{lemma} 
\label{lem:zero_trivial}
The transition waste incurred at Machine $n$ when transitioning from a set of tasks $\SNn$ to another set of tasks $\SNPn$ is zero if and only if $\SNn \subset \SNPn$ or $\SNn \supset \SNPn$. 
\end{lemma} 

In the next corollary, we show that when there are $N = L+1$ machines and one machine leaves or when there are $N = L$ machines and one machine joins, the transition waste is trivially zero, no matter which TASs the system are employing.

\begin{corollary} 
The transition waste when transitioning from an $(L,L,F)$-TAS to an $(L+1,L,F)$-TAS and vice versa are zero. 
\end{corollary} 
\begin{proof} 
Note that for an $(L,L,F)$-TAS $\S^L = (S^L_1,\ldots,S^L_L)$, we have $S^L_n = [[F]]$ for all $n \in [L]$. Therefore, $\SNn \supset \SNPn$. By Lemma~\ref{lem:zero_trivial}, the corollary follows.    
\end{proof} 

We henceforth assume that $N > L$ when one machine joins and $N > L+1$ when one machine leaves. 
First, we consider the case of one machine \textit{joining}. Theorem~\ref{thm:joining} establishes the transition waste of the cyclic task allocation scheme. 

\begin{theorem} 
\label{thm:joining}
The transition waste when transitioning from a cyclic \NT~$\SCN$ to a cyclic \NPOT~$\SCNPO$ (defined in \eqref{eq:cyclic}) is given below (assuming $N > L$).
\[
W(\SCN \to \SCNPO) = 
\dfrac{N-1}{N+1}F.
\]
\end{theorem}  
\begin{proof} 
We prove this theorem by performing a direct computation of the symmetric difference of the sets of tasks allocated before and after the transition.
Suppose Machine $N+1$ joins the computation.
According to \eqref{eq:cyclic}, we have
\[
\SCN = (S^N_1,\ldots,S^N_N) \text{ and }
\SCNPO = (S^{N+1}_1,\ldots,S^{N+1}_{N+1}),
\]
where for $n \in [N]$,
\[
\SNn = \left[(n-1)\frac{F}{N}, (n-1)\frac{F}{N} + \frac{LF}{N}-1\right] \pmod F,
\]
and for $n \in [N+1]$,
\[
\SNPOn \hspace{-2pt}=\hspace{-2pt} \left[\hspace{-2pt}(n-1)\frac{F}{N+1}, (n\hspace{-2pt}-\hspace{-2pt}1)\frac{F}{\hspace{-2pt}N+\hspace{-2pt}1}\hspace{-2pt} +\hspace{-2pt} \frac{LF}{N\hspace{-2pt}+\hspace{-2pt}1}\hspace{-2pt}-\hspace{-2pt}1\right]\hspace{-7pt} \pmod F. 
\]
We now apply Lemma~\ref{lem:sym} to find the symmetric difference of $\SNn$ and $\SNPOn$ for every $n \in [n]$. We write 
\[
S = \SNPOn = [a,b]\pmod F, \quad T = \SNn=[c,d]\pmod F 
\] 
and can verify that all assumptions of Lemma~\ref{lem:sym}~(a) are satisfied. Indeed, since $N > L$ and $N \geq n \geq 1$, we have 
\[
0 \leq a=(n-1)\frac{F}{N+1} \leq c=(n-1)\frac{F}{N} < F,
\]
\[
0 < |S| = \frac{LF}{N+1} < F,\quad 0 < |T| = \frac{LF}{N} < F, 
\]
\[
c-a = (n-1)\frac{F}{N(N+1)} < \frac{LF}{N+1}=|S|,
\]
\[
|S| = \frac{LF}{N+1} < (n-1)\frac{F}{N(N+1)} + \frac{LF}{N} = (c-a) + |T| < F.
\]
Therefore, by Lemma~\ref{lem:sym}~(a), 
\[
\begin{split}
|\SNn \Delta \SNPOn| &= 2(c-a) + (|\SNn|-|\SNPOn|)\\
&= \frac{2(n-1)F}{N(N+1)} + \frac{LF}{N(N+1)}\\ &= \frac{2(n-1)F}{N(N+1)} + \DNNPO.
\end{split}
\]
Thus, the transition waste incurred at Machine $n$ is
\[
W(\SNn\to \SNPOn) = |\SNn \Delta \SNPOn| - \DNNPO = \frac{2(n-1)F}{N(N+1)}.
\]
Finally, the transition waste when transitioning from $\SCN$ to $\SCNPO$ is
\[
\begin{split}
W(\SCN \to \SCNPO) &= \sum_{n\in [N]}W(\SNn\to \SNPOn)\\
&= \sum_{n\in [N]}\frac{2(n-1)F}{N(N+1)} = \frac{N-1}{N+1}F,
\end{split}
\]
as desired. 
\end{proof} 

We now turn to the slightly more involved case when one machine leaves the computation. When Machine $n^* \in [N]$ leaves, for the ease of notation, we assume the system transitions to the cyclic TAS 
\[
\SCNMO = \{S^{N-1}_1,\ldots,S^{N-1}_{n^*-1},S^{N-1}_{n^*+1},\ldots,S^{N-1}_N\}, 
\]
where for $n < n^*$,
\[
S^{N-1}_n = \bigg[(n-1)\frac{F}{N-1},(n-1)\frac{F}{N-1}+\frac{LF}{N-1}-1\bigg]\hspace{-10pt}\pmod F,
\]  
and for $n > n^*$, 
\[
S^{N-1}_n = \bigg[(n-2)\frac{F}{N-1},(n-2)\frac{F}{N-1}+\frac{LF}{N-1}-1\bigg]\hspace{-10pt}\pmod F. 
\]

\begin{lemma} 
\label{lem:leave_1}
Suppose that Machine $n^* \in [N]$ leaves and the system transitions from a cyclic \NT~$\SCN$ to a cyclic \NMOT~$\SCNMO$ (defined in \eqref{eq:cyclic}). The transition waste incurred at Machine $n$ for $n < n^*$ is (assuming $N > L+1$)
\[
W_{n^*}(\SNn \to \SNMOn) = 
\dfrac{2(n-1)F}{N(N-1)}.
\]
\end{lemma} 
\begin{proof} 
The proof is the same as that of Theorem~\ref{thm:joining}, whereby Lemma~\ref{lem:sym}~(a) is applied to $S = \SNn$ and $T = \SNMOn$.
\end{proof} 

\begin{lemma} 
\label{lem:leave_2}
Suppose that Machine $n^* \in [N]$ leaves and the system transitions from a cyclic \NT~$\SCN$ to a cyclic \NMOT~$\SCNMO$ (defined in \eqref{eq:cyclic}). The transition waste incurred at Machine $n$ for $N \geq n \geq n^*+1$ is given below (assuming $N > L+1$).

If $n^*\geq N-L$ then $W_{n^*}(\SNn \to \SNMOn) = 0$. 

If $n^* < N-L < n$ then $W_{n^*}(\SNn \to \SNMOn) = 0$.

If $n^* < n \leq N-L$ then 
\[
W_{n^*}(\SNn \to \SNMOn) = 
\dfrac{2(N-L-n+1)F}{N(N-1)}.
\]
\end{lemma} 
\begin{proof} 
\hoang{See Appendix~\ref{app:proof_leave_2}.}
\end{proof}

\hoang{Theorem~\ref{thm:leaving} dertermines the transition waste for the cyclic task allocation scheme when one machine \textit{leaves}.}

\begin{theorem} 
\label{thm:leaving}
The transition waste when Machine $n^* \in [N]$ leaves and the system transitions from a cyclic \NT~$\SCN$ to a cyclic \NMOT~$\SCNMO$ (defined in \eqref{eq:cyclic}) is given as follows (assuming $N > L+1$). 

If $n^*< N-L$, $W_{n^*}(\SCN\to \SCNMO)$ is\vspace{-5pt}
\[
\big((n^*-1)(n^*-2)+(N-L-n^*)(N-L-n^*+1)\big)\frac{F}{N(N-1)}.\vspace{-5pt}
\]

If $n^*\geq N-L$, $W_{n^*}(\SCN\to \SCNMO)$ is\vspace{-5pt}
\[
(n^*-1)(n^*-2)\frac{F}{N(N-1)}.\vspace{-5pt}
\]

Averaging $n^*$ over $[N]$, the averaged transition waste when one machine leaves in the cyclic ETAS is
\begin{multline*}
W_{\sf{avg}}(\SCN\to \SCNMO)\\
=
\bigg(\frac{N-2}{3N} + \frac{(N-L-1)(N-L)(N-L+1)}{3(N-1)N^2} \bigg)F.
\end{multline*}
\end{theorem} 
\begin{proof} 
If $n^*<N-L$, by Lemma~\ref{lem:leave_1} and Lemma~\ref{lem:leave_2}, we have
\begin{multline*}
W_{n^*}(\SCN\to \SCNMO)\\ = \sum_{n=1}^{n^*-1}\frac{2(n-1)F}{N(N-1)} + \sum_{n=n^*+1}^{N-L}\frac{2(N-L-n+1)F}{N(N-1)} +\hspace{-8pt} \sum_{n=N-L+1}^N 0\\
= \big((n^*-1)(n^*-2)+(N-L-n^*)(N-L-n^*+1)\big)\frac{F}{N(N-1)}.
\end{multline*}
Similarly, when $n^* \geq N-L$, we obtain
\[
\begin{split}
W_{n^*}(\SCN\to \SCNMO) &= \sum_{n=1}^{n^*-1}\frac{2(n-1)F}{N(N-1)} + \sum_{n=n^*+1}^{N} 0\\ &= (n^*-1)(n^*-2)\frac{F}{N(N-1)}.
\end{split}
\]
Averaging $W_{n^*}(\SCN\to \SCNMO)$ over all $n^*\in [N]$ we obtain the stated formula for $W_{\sf{avg}}(\SCN\to \SCNMO)$.
\end{proof} 

\subsection{Shifted Cyclic Scheme Achieving Optimal Transition Waste}
\label{subsec:TW_shifted_cyclic}

From Theorem~\ref{thm:joining} and Theorem~\ref{thm:leaving}, the transition waste incurred across all existing machines in the cyclic ETAS proposed in~\cite{Yang_etal_2019} is $\frac{N-1}{N+1}F \approx F$ or $(\frac{N-2}{3N}+\cdots)F \approx \frac{F}{3}$ tasks when a machine joins or leaves, respectively. 
In this section, we show that by applying a calculated shift, we can significantly reduce the transition waste of the cyclic ETAS. 

As mentioned earlier, the updated TAS used by the cyclic ETAS~\cite{Yang_etal_2019} (see Section~\ref{sec:pre}) only depends on the number of machines available and not on the current TAS, which is one reason that leads to the scheme's poor transition waste. We now generalize the cyclic TAS to \emph{shifted} cyclic TAS in order to allow a more adaptive transition that takes into account the current TAS. 

\begin{definition}[Shifted cyclic task allocation] 
\label{def:shift}
For $\delta \in [[F]]$, a $\delta$-shifted cyclic \NT~is given as follows.
\[
\SDCN = (S^N_1,\ldots,S^N_N),\]
where for $n \in [N]$,
\[
\SNn = \left[(n-1)\frac{F}{N}+\delta, (n-1)\frac{F}{N} + \frac{LF}{N}-1+\delta\right]\hspace{-10pt} \pmod F.
\]
\end{definition}    

Note that there are $F$ different shifted TASs possible corresponding to $F$ different values of $\delta$. When $\delta = 0$, the shifted cyclic TAS reduces to an ordinary cyclic TAS (Section~\ref{sec:pre}).  

Given that the system transitions from an $\delta'$-shifted cyclic \NT~to a $\delta$-shifted cyclic \NPT, the question of interest is to determine $\delta$ that leads to a minimum transition waste. 
We note here that the master machine can always exhaustively examine all possible $F$ shifted schemes and find the one with the smallest waste. However, this will take the master roughly $LF^2 = FN\frac{LF}{N}$ operations, which is time-consuming for large $F$. \emph{Our contribution} is to derive the explicit formula of an \emph{optimal shift}, which results in the \emph{minimum waste} among all $F$ shifted schemes. 
We first tackle the case of one machine joining and then argue that the case of one machine leaving follows by symmetry.  

\hoang{Theorem~\ref{thm:shifted_joining} computes the transition waste for a particular shifted cyclic task allocation scheme when one machine \textit{joins}. The amount of shift $\delta-\delta'$ given in the theorem will be partially proved to be optimal in Theorem~\ref{thm:optimal}.} 

\begin{theorem} 
\label{thm:shifted_joining}
The transition waste when transitioning from a $\delta'$-shifted cyclic \NT~$\SDPCN$ to a $\delta$-shifted cyclic \NPOT~$\SDCNPO$ with $\delta\hspace{-2pt} =\hspace{-2pt} \delta'\hspace{-2pt} +\hspace{-2pt} \lfloor \frac{N+L-1}{2} \rfloor \frac{F}{N(N+1)}$ is
\[
W(\SDPCN\hspace{-3pt} \to\hspace{-1pt} \SDCNPO)\hspace{-2pt} = \hspace{-2pt}
\begin{cases}
\frac{(N-L-1)(N-L+1)F}{2N(N+1)},\hspace{-5pt} &\text{for odd } N - L,\\
\frac{(N-L)^2F}{2N(N+1)},\hspace{-5pt} &\text{for even } N - L.\\
\end{cases}
\]
\end{theorem}

\begin{proof}[Proof of Theorem~\ref{thm:shifted_joining}]
\hoang{See Appendix~\ref{app:proof_shifted_joining}.}
\end{proof} 

\hoang{By comparing the formulas derived in Theorem~\ref{thm:joining} and Theorem~\ref{thm:shifted_joining},} we deduce that the transition waste of the proposed shifted cyclic TAS when a machine joins is improved over that of the ordinary cyclic TAS (\cite{Yang_etal_2019}) by a considerable factor of approximately $\frac{2N^2}{(N-L)^2}$, which is 8X when $L \approx N/2$. The improvement becomes even more significant when $L$ gets closer to $N$, e.g., in the order of $N^2$ when $N-L$ is small.

Theorem~\ref{thm:shifted_leaving} determines the transition waste for a particular shifted cyclic task allocation scheme when one machine \textit{leaves}. The amount of shift $\delta-\delta'$ given in the theorem will be proved to be optimal in Theorem~\ref{thm:optimal} under a divisibility condition. 

\begin{theorem} 
\label{thm:shifted_leaving}
The transition waste when transitioning from a $\delta'$-shifted cyclic \NT~$\SDPCN$ to a $\delta$-shifted cyclic \NMOT~$\SDCNMO$ with $\delta = \delta' + \big((N-n^*)- \lfloor \frac{(N+L-2)}{2} \rfloor\big) \frac{F}{N(N-1)}$, where Machine $n^*$ leaves, is
\[
W(\SDPCN\hspace{-3pt} \to\hspace{-1pt} \SDCNMO)\hspace{-2pt} = \hspace{-2pt}
\begin{cases}
\frac{(N-L-1)^2F}{2N(N-1)},\hspace{-5pt} &\text{for odd } N - L,\\
\frac{(N-L)(N-L-2)F}{2N(N-1)},\hspace{-5pt} &\text{for even } N - L.\\
\end{cases}
\]
\end{theorem} 
\begin{proof} 
The proof works by symmetry. By treating Machine $n^*$ that leaves as the machine that joins the system in Theorem~\ref{thm:shifted_joining} and replacing $N$ by $N-1$, we obtain the claimed formula for the transition wastes. Note that because the task sets can be cyclically shifted along the circle of integers mod $F$, the index of the machine that leaves does not matter. This phenomenon, however, does not apply to the ordinary cyclic ETAS.  
\end{proof} 

Although we are able to show the optimality of our shifted cyclic ETASs only when the parameter $\delta$ satisfies a certain divisibility property (Theorem~\ref{thm:optimal}), we believe the optimality holds for every $\delta$, which was supported by an exhaustive search over small values of $L$ and $N$. 

\begin{theorem} 
\label{thm:optimal}
The transition wastes stated in Theorem~\ref{thm:shifted_joining} and Theorem~\ref{thm:shifted_leaving} are optimal among all choices of $\delta$-shifted cyclic TASs where $\frac{F}{N(N+1)}$ and $\frac{F}{N(N-1)}$ divide $\delta-\delta'$, respectively. 
\end{theorem} 
\begin{proof} 
By symmetry, we just need to prove this for the case of machines joining. 
We first derive a formula of the transition waste for every $\delta$ and then show that it is minimized within the specified range of $\delta$. See Appendix~\ref{sub:AppendixOptProof} for more details. 
\end{proof} 

\section{Zero-Waste Elastic Task Allocation Schemes}
\label{sec:zero}

The shifted cyclic ETAS developed in Section~\ref{subsec:TW_shifted_cyclic} is easy to implement and has a negligible computation overhead at the master machine. \hoang{Indeed, to coordinate a transition, the master just needs to inform each machine its updated index, the number of active machines, and the amount of shift required.} However, to maintain the cyclic structure, the transitions incur a nontrivial transition waste. \hoang{This quantity can scale linear in $F$, which is the maximum number of tasks each machine can take, and hence may significantly increase the computation overhead at each machine.} 
Moreover, high transition wastes also mean more new tasks than necessary must be downloaded if each machine does not already store all the tasks from the beginning, which leads to higher communication overhead. 

This drawback of the (shifted) cyclic ETAS motivated us to investigate elastic task allocation schemes with \emph{zero} transition wastes. Our key findings include a necessary and sufficient condition for the existence of a zero-waste transition from an \NT~to an \NPT~based on the famous Hall's marriage theorem and a construction of zero-waste ETAS based on finite geometry. \vspace{-5pt}

\subsection{Zero-Waste Transition When One Machine Joins}

By Lemma~\ref{lem:zero_trivial}, the transition waste incurred at Machine $n$ when transitioning from the set of tasks $\SNn$ to another set $\SNPn$ is zero if and only if $\SNn \subset \SNPn$ or vice versa. 
It turns out that if \hoang{the elastic events consist of only} machines joining then it is easy to achieve zero-waste transitions. 

\begin{proposition}
There always exists a zero-waste transition from an \NT~to an \NPOT. 
\end{proposition}
\begin{proof} 
To achieve a zero-waste transition when Machine $N+1$ joins, each existing machine (from $1$ to $N$) can simply choose a subset of $\frac{LF}{N(N+1)}$ tasks to pass to Machine $N+1$, which will then have in total $N\frac{LF}{N(N+1)}$ tasks. The requirement is to have these $N$ sets disjoint. We can achieve this by letting each machine $n$ from $1$ to $N$ choose an arbitrary subset of $\SNn$ of size $\frac{LF}{N(N+1)}$ that does not intersect any sets chosen by previous machines so far. This is always possible because Machine $n$ has enough tasks in its set to do the selection:   
\[
|\SNn| = \frac{LF}{N} \geq (n-1)\frac{LF}{N(N+1)}+\frac{LF}{N(N+1)}.\vspace{-5pt}
\]
\hoang{This completes the proof.}
\end{proof}

Note that this proposition is a stand-alone result and will not be used in the rest of the paper.  \vspace{-5pt}

\subsection{Zero-Waste Transition When One Machine Leaves}

The case of one machine leaving, say Machine $n^*$, is more challenging. Note that to achieve a zero-waste transition, due to Lemma~\ref{lem:zero_trivial}, it is necessary and sufficient to let other machines keep their current sets of tasks while reallocating the tasks from the leaving machine to them (so that $\SNn \subset \SNMOn$). Reallocating one task from Machine $n^*$ to a machine $n$ corresponds to selecting one edge in the \emph{transition graph} (Definition~\ref{def:transition_graph} below). 
Note that when the transition happens, both $L$ and $F$ are fixed. This means that each task from the leaving machine needs to be reallocated to exactly \textit{one} active machine to maintain the $L$-Redundancy, which requires that each task is performed by exactly $L$ machines (see Definition~\ref{def:TAS}).
We will see later that reallocating all tasks turns out to correspond to a ``matching'' in that graph (Lemma~\ref{lem:matching}).  

\begin{definition} 
\label{def:transition_graph}
Given an \NT~$\SN=(S^N_1,\ldots,S^N_N)$, the transition graph $\Gns$ is the bipartite graph with vertex set $\Uns \cup \Vns$, where $\Uns = [N]\setminus \{n^*\}$ and $\Vns = \SNns$ and there is an edge $(u,v)$, $u \in \Uns$, $v \in \Vns$, if and only if $v \in \overline{S^N_u}\define[[F]]\setminus S^N_u$. 
\end{definition} 

Note that the set $\Vns$ of the transition graph represents the tasks from the leaving machine $n^*$ that need to be reallocated to other machines, while an edge $(u,v)$ implies that the task $v \in \Vns$ can be taken over by Machine $u$, i.e., this machine was not allocated this task before the transition. An example of such a graph is given below. \vspace{-5pt}

\begin{figure}[!htb]
\centering
\includegraphics[scale=0.68]{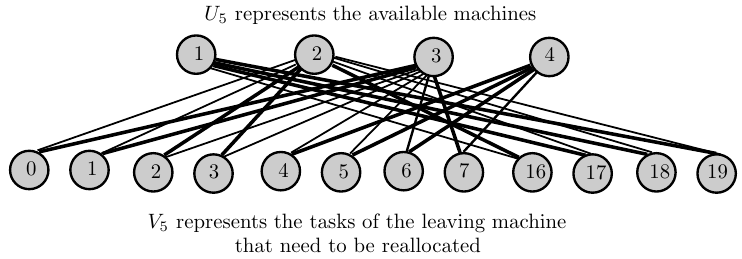}
\caption{Illustration of the transition graph $\G_5$ in Example~\ref{ex:transition_graph}. An edge $(u,v)$ means the task $v$ from the leaving machine can be taken over by Machine $u$ because Machine $u$ was not allocated this task before the transition.}
\label{fig:transition_graph}
\vspace{-10pt}
\end{figure}

\begin{example} 
\label{ex:transition_graph}
When $N = 5$, $L=3$, and $F=20$, we consider the \NT~given in Fig.~\ref{fig:toy_example}~(a), $\S^5 = (S^5_1,\ldots,S^5_5)$, where $S^5_1 = \{0,\ldots,11\}$, $S^5_2 = \{4,\ldots,15\}$, $S^5_3 = \{8,\ldots,19\}$, $S^5_4 = \{0,\ldots,3,8,\ldots,19\}$, and $S^5_5 = \{0,\ldots,7,16,\ldots,19\}$. Suppose that Machine $5$ leaves, i.e., $n^* = 5$. Then the transition graph $\G_5$ is illustrated in Fig.~\ref{fig:transition_graph}.  \vspace{-5pt}
\end{example} 

A subset $\M$ of edges of a bipartite graph $\G$ with vertex set $(U,V)$ is referred to as a \emph{perfect $\Delta$-matching} of $\G$ if each vertex in $V$ is incident to precisely one edge in $\M$ while each vertex in $U$ is incident to precisely $\Delta$ edges in $\M$. 
\vspace{-5pt}

\begin{lemma} 
\label{lem:matching}
There exists a zero-waste transition from an \NT~$\SN$ to an \NMOT~$\SNMO$ when Machine $n^*$ leaves if and only if the transition graph $\Gns$ admits a perfect $\DNNMO$-matching. 
\end{lemma} 
\begin{proof} 
Recall that due to Lemma~\ref{lem:zero_trivial}, the transition has a zero-waste if and only if $\SNMOn \subset \SNn$ for every $n \in [N]\setminus \{n^*\}$. This means that we need to reallocate tasks left over by Machine $n^*$ to other $N-1$ machines by adding these new tasks to the existing task sets of these machines.  

It is evident that a way to reallocate $\frac{LF}{N}$ tasks from Machine $n^*$ to $N-1$ other machines corresponds precisely to a perfect $\DNNMO$-matching of the transition graph $\Gns$: each task, which corresponds to a vertex $v \in \Vns$, is reallocated to exactly \emph{one} machine, which corresponds to a vertex $u \in \Uns$; moreover, each machine is allocated precisely $\DNNMO = \frac{LF}{N(N-1)}$ new tasks, which shows that each vertex $u$ is incident to precisely $\DNNMO$ edges while each vertex $v$ is incident to exactly one edge in the matching.
\end{proof} 

For instance, the zero-waste transition presented in Fig.~\ref{fig:toy_example}~(a)(c) corresponds to the following perfect $3$-matching of $\G_5$ (thicker edges in Fig.~\ref{fig:transition_graph}): 
\begin{multline*}
\M = \{(1,17), (1,18), (1,19), (2,2), (2,3),\\ (2,16), (3,0), (3,1), (3,7), (4, 4), (4,5), (4,6)\}.
\end{multline*}
Based on this matching, each machine $1$, $2$, $3$, and $4$ is allocated three new tasks from the leaving Machine 5. Moreover, every task from Machine 5, i.e., $\{0,1,2,3,4,5,6,7,16,17,18,19\}$, is reallocated to exactly one machine. 

The following lemma is a straightforward corollary of Hall's marriage theorem.

\begin{lemma}
\label{lem:Hall_delta}
A bipartite graph $\G$ with the vertex set $U \cup V$ has a perfect $\DNNMO$-matching if and only if the inequality
\begin{equation}
\label{eq:Hall_inequality_variation}
|\cup_{n \in J}\Gamma_{\G}(n)| \geq |J|\DNNMO,
\end{equation}
holds for every nonempty set $J \subseteq U$, where $\Gamma_{\G}(n)$ denotes the set of neighbors of $n$ in $\G$. 
\end{lemma} 
\begin{proof} 
The celebrated Hall's marriage theorem~\cite{Hall} states that a bipartite graph $\G$ with the vertex set $(U,V)$ has a perfect matching (or, perfect $1$-matching, in our notation), if and only if for every nonempty set $J \subseteq U$, it holds that $
|\cup_{n \in J}\Gamma_{\G}(n)| \geq |J|$, where $\Gamma_{\G}(n)$ denotes the set of neighbors of $n$ in $\G$. 
By duplicating each vertex of $U$ and its incident edges $\DNNMO$ times and applying Hall's theorem to the resulting bipartite graph, we deduce that $\G$ has a perfect $\DNNMO$-matching if and only if \eqref{eq:Hall_inequality_variation} holds for every nonempty set $J \subseteq U$. 
\end{proof} 

As a corollary of Lemma~\ref{lem:matching} and Lemma~\ref{lem:Hall_delta}, we obtain a necessary and sufficient condition for the existence of a zero-waste transition when one particular machine leaves. 

\begin{corollary}
\label{cr:one_leave}
There exists a zero-waste transition from an \NT~$\SN$ to an \NMOT~$\SNMO$ when Machine $n^*$ leaves if and only if the following inequality holds for every nonempty set $J \subseteq [N]\setminus \{n^*\}$. 
\begin{equation}
\label{eq:Hall_inequality}
|(\cup_{n \in J} \overline{\SNn}) \cap \SNns| \geq |J|\DNNMO.
\end{equation} 
\end{corollary}
\begin{proof} 
The conclusion is straightforward from Lemma~\ref{lem:matching} and Lemma~\ref{lem:Hall_delta} and the following observation: by the definition of the transition matrix $\Gns$, the set of neighbours of a vertex $n \in \Uns = [N]\setminus \{n^*\}$ in $\Gns$ is 
$\Gamma_{\Gns}(n) = \overline{\SNn} \cap \SNns$. 
\end{proof} 

Theorem~\ref{thm:Hall_like} provides a necessary and sufficient condition for the existence of a zero-waste transition from an \NT~to an \NMOT~no matter which machine leaves. Essentially, it states that as long as the sets of tasks of different machines do not overlap too much then there exists a zero-waste transition.  
Recall that $\DNNMO = \frac{LF}{N(N-1)}$.

\begin{theorem} 
\label{thm:Hall_like}
There exists a zero-waste transition from an \NT~$\SN\hspace{-5pt}=\hspace{-3pt}(S^N_1\hspace{-3pt},\ldots,S^N_N)$ to an \NMOT{} when Machine $n^*$ leaves for every $n^* \in [N]$ if and only if
\begin{equation} 
\label{eq:main}
|\cap_{n \in I} \SNn| \leq (N-|I|)\DNNMO,
\end{equation} 
for every nonempty set $I \subseteq [N]$.
Moreover, such a transition can be found in time $\O\big((N-1+\frac{LF}{N})(N-1)F(1-\frac{L}{N})\big)$.
\end{theorem}
\begin{proof} 
Let $\Gns$ be the transition graph of an \NT~$\SN$ with the vertex set $(\Uns,\Vns)$ (Definition~\ref{def:transition_graph}).
By Corollary~\ref{cr:one_leave}, it suffices to show that the inequality \eqref{eq:Hall_inequality} holds for every nonempty set $J \subseteq [N]\setminus \{n^*\}$ and for every $n^* \in [N]$ if and only if \eqref{eq:main} holds for every nonempty set $I \subseteq [N]$.

Suppose that \eqref{eq:Hall_inequality} holds as stated. Note that
\[
\begin{split}
|(\cup_{n \in J} \overline{\SNn}) \cap \SNns| &= |\overline{\cap_{n \in J} \SNn} \cap \SNns|
= |\SNns \setminus \cap_{n \in J} \SNn|\\
&= |\SNns| - |\cap_{n \in J \cup \{n^*\}} \SNn|.
\end{split}
\] 
Therefore, \eqref{eq:Hall_inequality} is equivalent to 
\[
|\cap_{n \in J \cup \{n^*\}} \SNn| \leq |\SNns| - |J|\DNNMO.
\]
Setting $I = J \cup \{n^*\}$, this is also equivalent to 
\[
\begin{split}
|\cap_{n \in I} \SNn| &\leq |\SNns| - (|I|-1)\DNNMO\\
&= (N-1)\DNNMO - (|I|-1)\DNNMO\\
&= (N-|I|)\DNNMO.
\end{split}
\]
Note that as $n^*$ varies over $[N]$ and $J$ varies over all nonempty subsets of $[N]\setminus \{n^*\}$, $I = J \cup \{n^*\}$ varies over all subsets of $[N]$ of size at least two. Furthermore, \eqref{eq:main} holds trivially (with equality) when $|I|=1$. Therefore, \eqref{eq:main} holds for all nonempty sets $I \subseteq [N]$. Hence, we settle the \emph{only if} direction. As all steps are equivalent transformations, the \emph{if} direction is also true. The complexity of finding a zero-waste transition comes from that of a network flow algorithm~\cite{AhujaMagnantiOrlin} employed to find a perfect matching for $\Gns$. This completes the proof.     
\end{proof} 

Theorem~\ref{thm:Hall_like} provides us with an important insight: to make transitions with zero waste possible, we should assign to machines sets of tasks with small overlaps. This will be crucial in our construction of an ETAS with zero transition waste in the next section.

\subsection{A Zero-Waste Elastic Task Allocation Scheme}
So far we have discussed the case of a single machine leaving or joining. The more challenging question is how to allow a (possibly infinite) chain of such elastic events while guaranteeing zero-waste transitions. More specifically, we are interested in establishing a \emph{zero-waste range} $\Nmm \subset [L,F]$ where the system can start with any number $N_0$ of machines, $N_0 \in [\Nmi, \Nm]$, and then can transition \emph{with zero wastes} an arbitrary number of times \emph{within this range}, one machine leaving or joining at a time. 
We show the existence of a handful of such ranges in Theorem~\ref{thm:ZWR} and Corollary~\ref{cr:ZWR}. 
We first need a formal definition of a zero-waste range.

\begin{definition}[Zero-waste range] 
\label{def:ZWR}
Given $L$ and $F$, a range $\Nmm$, where $L \leq \Nmi \leq \Nm \leq F$ is called an $(L,F)$-zero-waste range ($(L,F)$-ZWR) if for every $N_0 \in \Nmm$ there exists an $(N_0,L,F)$-ETAS $(\SNz,\T)$ (see Definition~\ref{def:ETAS}) where the transition algorithm $\T$ incurs a zero waste whenever the transition is within the range $\Nmm$.  
\end{definition} 

Note that $\Nmi$ and $\Nm$ are usually functions of $L$ and~$F$. Also, the transition algorithm $\T$ mentioned in Definitions~\ref{def:ETAS} and~\ref{def:ZWR} can be applied repeatedly to enable a chain of transitions within $\Nmi$ and $\Nm$ machines by adding or removing one machine at a time.
It turns out that if we can construct an $(N_0,L,F)$-ETAS $(\SNz,\T)$ so that $\T$ incurs a zero transition waste within $\Nmm$ for \emph{some} $N_0 \in \Nmm$ then we can also construct an $(N'_0,L,F)$-ETAS satisfying the same property for \emph{every} $N'_0 \in \Nmm$, i.e., $\Nmm$ is an \lfr. In particular, we show that this claim is true when $N_0 = \Nm$.

\begin{lemma} 
\label{lem:Nmax}
If there exists an $(\Nm, L, F)$-ETAS $(\S^{\Nm},\T)$ so that $\T$ always incurs a zero transition waste for every possible chain of $\Nm-\Nmi$ transitions from $\Nm$ to $\Nmi$ machines (machines leaving only) then $\Nmm$ is an \lfr.
\end{lemma}

Before proving this lemma, we need the concept of a \emph{transition tree}, which keeps track of all the possible \emph{states} the system can be at and the transitions leading to them from the original state, where a state consists of the list of machines available and the corresponding TAS.
The transition tree is, in fact, an explicit way to represent an ETAS.  

\begin{figure}[!htb]
\centering
\includegraphics[scale=1]{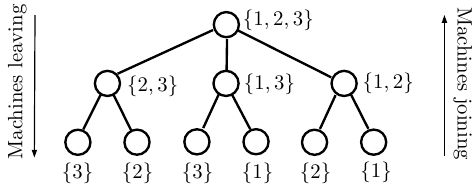}
\caption{Illustration of a transition tree when $\Nmi = 1$ and $\Nm = 3$. The set of available machines is given at each node (we omit the TAS associated with each node).}
\label{fig:ex_transition_tree}
\end{figure}

\begin{definition}[Transition tree] 
Given an $(\Nm, L, F)$-ETAS $(\S^{\Nm},\T)$ satisfying the assumption of Lemma~\ref{lem:Nmax}, the corresponding transition tree $\mathfrak{T}$ is a rooted tree created as follows.
The root node of the tree consists of the set $[\Nm]$ and the corresponding $(\Nm,L,F)$-TAS. Other nodes can be created in a recursive manner. Suppose that a node $u$ is already created that consists of a set of indices $\I$ and an $(|\I|,L,F)$-TAS. If $|\I| > \Nmi$, the $|\I|$ child nodes of $u$ can be created as follows. Each child node $v$ corresponds to the removal of one machine indexed by $n^*\in \I$ and consists of the list $\J \define \I \setminus \{n^*\}$ and a $(|\J|,L,F)$-TAS obtained by applying the transition algorithm $\T$ to the $(|\I|,L,F)$-TAS of $u$. 
\end{definition}  

For instance, when $\Nm = 3$, $\Nmi = L = 1$, we have a transition tree illustrated in Fig.~\ref{fig:ex_transition_tree}. 
\begin{proof}[Proof of Lemma~\ref{lem:Nmax}]
Based on the transition tree, it is easy to see that once the system can start from an $(\Nm,L,F)$-TAS and transition with zero wastes down to an $(\Nmi, L,F)$-TAS in all possible ways then we can also start from any intermediate $(N_0,L,F)$-TAS, $N_0 \in \Nmm$, and transition with zero wastes within this range. Indeed, if one machine leaves and the system is currently at a state corresponding to a node in the tree, then it can transition to a child node depending on which node is leaving. Vice versa, if one machine joins, the system can transition to the state stored at the parent node.  
\end{proof}

\begin{remark}[Overhead incurred by the transition tree]
As shown in the proof of Lemma~\ref{lem:Nmax}, the transition tree is used to keep track of all zero-waste transitions possible within the range $\Nmm$. The entire tree can be created once by the master machine before the computation session starts or can be created on the fly. The tree has height $\Nm-\Nmi$ and a total of 
$
1+\sum_{h = 1}^{\Nm-\Nmi} \prod_{i=0}^{h-1} (N-i)
$
nodes, which is in the order of $N!$. To create a child node, a network flow algorithm is invoked to find the zero-waste transition (however, the computation required becomes lighter when it gets closer to the leaves). The creation and storage of the transition tree incurs significant storage and computation overheads at the master node, and therefore, using the tree is beneficial when we have relatively small $N$ and $F$ and intensive tasks so that having zero transition waste pays off. Maintaining a zero-waste ETAS with lower overheads remains an open question for future research. 
\end{remark}

Based on Lemma~\ref{lem:Nmax}, we now describe our construction of \lfr s~based on the so-called \emph{symmetric configurations} from combinatorial designs.
  
\begin{definition}[Configuration~\cite{colbourn2006handbook}] 
\label{def:conf}
A $(v, b, k, r)$-configuration is an incident structure of $v$ points and $b$ lines such that 
\begin{itemize}
	\item each line contains $k$ points,
	\item each point lies on $r$ lines, and
	\item two different points are connected by at most one line.
\end{itemize}
If $v = b$ and, hence, $r = k$, the configuration is symmetric, denoted by $(v,k)$-configuration. 
\end{definition} 

The famous Fano plane is a $(7,3)$-configuration with seven points $\{1,2,\ldots,7\}$ and seven lines: $\{1,2,3\}$, $\{1,4,5\}$, $\{1,6,7\}$, $\{2,4,6\}$, $\{2,5,7\}$, $\{3,5,6\}$, and $\{3,4,7\}$ (Fig.~\ref{fig:Fano}).  \vspace{-10pt}
\begin{figure}[!htb]
\centering
\includegraphics[scale=0.35]{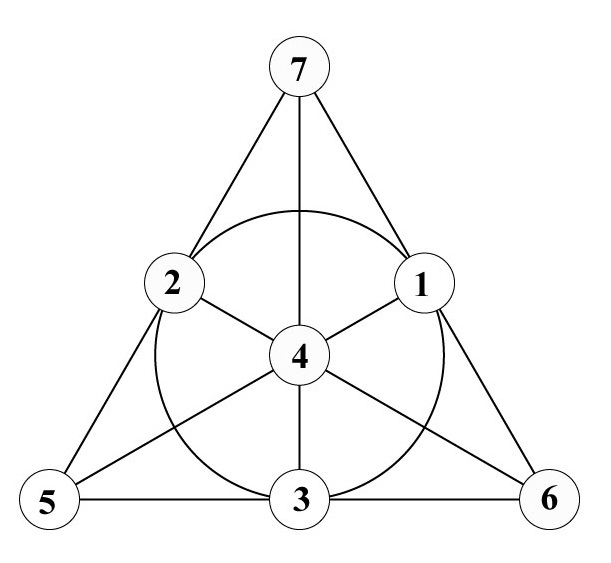}
\caption{A Fano plane with seven points and seven lines.}
\label{fig:Fano}
\end{figure}

We first show that an $(\Nm,L)$-configuration can be used to construct an $(\Nm,L,F)$-TAS with small pairwise overlaps and then present a method to establish an $\Nmm$-zero-waste range from such a TAS. Essentially, points correspond to tasks while lines correspond to sets of tasks.  As there are $\Nm$ points and $F$ tasks, it is natural to associate each point with $F/\Nm$ tasks.\\

\textbf{Construction 1.} Suppose that $\Nm$ divides $F$ and $\B=\{B_1,\ldots,B_{\Nm}\}$ is the set of $\Nm$ lines of an $(\Nm,L)$-configuration. An $(\Nm,L,F)$-TAS $\SNm$ can be constructed as follows. 
First, partition $[[F]]=\{0,\ldots,F-1\}$ into $\Nm$ equal sized parts $F_1,\ldots,F_{\Nm}$. Then for each $n \in [\Nm]$ we assign to Machine $n$ the tasks indexed by the parts $F_p$'s corresponding to all points $p$ in the line $B_n$. In other words, we set $\SNmn := \cup_{p \in B_n} F_p$, for every $n \in [\Nm]$.\\

For instance, when there are $\Nm = 7$ machines, $L = 3$, and $F = 14$ tasks, we first partition $[[F]]$ in to seven parts: 
\[
\begin{split}
&F_1 = \{0,1\}, F_2 = \{2,3\}, F_3 = \{4,5\},\\ &F_4 = \{6,7\}, F_5 = \{8,9\}, F_6 = \{10,11\}, F_7 = \{12,13\}.
\end{split}
\]
Then, using the $(7,3)$-configuration (the Fano plane) in Construction~1, we obtain a $(7,3,14)$-TAS, represented by Fig.~\ref{fig:Fano_TAS}. For instance, Machine~1 is allocated the task set $S^7_1 = \{0,1,\ldots,5\} =F_1\cup F_2 \cup F_3$, while Machine~2 has the task set $S^7_2 = \{0,1,6,7,8,9\} = F_1 \cup F_4 \cup F_5$. Clearly, each task is performed by $L= 3$ machines and each machine performs $LF/\Nm=6$ tasks.

\begin{figure}[!htb]
	\tabcolsep=0.07cm
		\renewcommand{\arraystretch}{1.2}%
	\centering	
	\begin{tabular}{|c|c|c|c|c|c|c|}
	\hline      
	$S^7_1$ & $S^7_2$ & $S^7_3$ & $S^7_4$ & $S^7_5$ & $S^7_6$ & $S^7_7$\\
    \hline
    $\{0,1\}$ & $\{0,1\}$ & $\{0,1\}$ 	&			& 			& 			& \\
    \hline
    $\{2,3\}$ &       &     	& $\{2,3\}$ 	& $\{2,3\}$		&  			& \\
    \hline
    $\{4,5\}$ &       &     	& 	 		&  			& $\{4,5\}$ 	& $\{4,5\}$ \\
    \hline
          & $\{6,7\}$ &     	& $\{6,7\}$ 	& 			& 			& $\{6,7\}$ \\
    \hline
          & $\{8,9\}$ &     	&  			& $\{8,9\}$		& $\{8,9\}$		& \\
    \hline
          &       & $\{10,11\}$ & $\{10,11\}$	& 			& $\{10,11\}$ 	& \\
    \hline
          &       & $\{12,13\}$ &  			& $\{12,13\}$ 	&			& $\{12,13\}$\\
    \hline
     \end{tabular}
     \caption{A $(7,3,14)$-TAS constructed from the Fano plane. The table rows/columns correspond to the plane points/lines.} 
     \label{fig:Fano_TAS}
\end{figure}
Since every two lines in a configuration either don't intersect or intersect at only one point, the resulting TAS also has small pairwise intersections, which is crucial for our construction of a zero-waste range. 

\begin{lemma} 
\label{lem:configuration_TAS}
Construction~1 produces an \NmT{} where every two task sets intersect at at most $F/\Nm$ tasks. 
\end{lemma} 
\begin{proof} 
According to Construction~1, each set of task has size
\[
|\SNmn| = |B_n|\frac{F}{\Nm} = \frac{LF}{\Nm}.
\]
Moreover, as each point $p$ in the configuration belongs to exactly $L$ lines, each task also belongs to precisely $L$ task sets. Hence, the resulting $\SNm$ is indeed an \NmT. Moreover, since every two lines in the configuration intersect at at most one point, every two task sets $\SNmn$ and $\SN_{n'}$, $n \neq n'$, intersect at at most $F/\Nm$ tasks as claimed. 
\end{proof} 

Note that the \emph{expected} cardinality\footnote{Indeed, as each point in a set of size $F$ belongs to a random subset of size $\frac{LF}{N}$ with probability $\frac{L}{N}$, the probability that a point belongs to two independent random subsets of size $\frac{LF}{N}$ is $\frac{L^2}{N^2}$. This implies that the expected size of the intersection of the two random subsets of that size is $\frac{L^2}{N^2}F = \frac{L^2}{N}\frac{F}{N}$.} of the intersection of two \emph{random} subsets of cardinality $LF/N$ of $[[F]]=\{0,1,\ldots,F-1\}$ is $\frac{L^2}{N}\frac{F}{N}$, which is approximately $F/N$ for $L^2 \approx N$. 
Therefore, a random construction doesn't provide smaller (expected) pairwise intersections than Construction~1.


By Lemma~\ref{lem:configuration_TAS}, Construction~1
produces an initial \NmT~with small pairwise set overlaps. To show that $R$ machines can be removed one by one from this TAS with zero transition wastes, we first show that the pairwise intersections of the sets of intermediate TASs do not increase too much. Then, by using the pairwise intersection as an upper bound on the intersection of any set $I$ of task sets, $|I| \leq L$, we can guarantee that the intersections still satisfy the Hall-like condition in Theorem~\ref{thm:Hall_like}. As a consequence, zero-waste transitions will be possible within the range $[\Nm-R,\Nm]$.

\begin{theorem} 
\label{thm:ZWR} 
If there exists an $(\Nm,L)$-configuration then  there exists an \NmT~$\SNm=(S^{\Nm}_1,\ldots,S^{\Nm}_{\Nm})$ where
\[
|S^{\Nm}_n \cap S^{\Nm}_{n'}| \leq \frac{F}{\Nm},
\]
for every $n, n' \in [\Nm]$, $n \neq n'$. This leads to the existence of an $(L,F)$-zero-waste range $[\Nm-R,\Nm]$ where
\begin{equation} 
\label{eq:R}
R \define 1+\left\lfloor\frac{(3L\Nm-2\Nm-2L+1)-\sqrt{\Delta}}{4L-2}\right\rfloor
\end{equation} 
and 
\begin{equation} 
\label{eq:delta}
\Delta = L\Nm(L\Nm + 8L^2 - 16L + 6) + (2L-1)^2.
\end{equation}
We assume here that $N$ \hoang{divides} $F$ for every $N \in [\Nm-R,\Nm]$, \hoang{and $L\geq 1$ and $\Nm \geq 2$.} 
\end{theorem}   
\begin{proof} 
The first statement is due to Lemma~\ref{lem:configuration_TAS}. We now prove the second statement, assuming that there exists an \NmT~as specified.
Thanks to Lemma~\ref{lem:Nmax}, it suffices to show that for every $1 \leq r < R$, after removing any $r$ machines one after another, the resulting $(\Nm-r, L, F)$-TAS still admits a zero-waste transition when one more machine leaves. Equivalently, we aim to show that this TAS satisfies the Hall-like condition~\eqref{eq:main}.

Suppose that $r < R$ machines have been removed with $r$ zero-waste transitions and $\S^{\Nm-r} = (S^{\Nm-r}_1,\ldots,S^{\Nm-r}_{\Nm-r})$ is the resulting $(\Nm-r,L,F)$-TAS. Let $I$ be a nonempty subset of indices of $|I|$ machines among the remaining ones. Note that when $|I| = 1$ or $|I| > L$, the inequality~\eqref{eq:main} is trivially satisfied. Indeed, when $|I| = 1$, the equality is achieved. When $|I| > L$, as each task cannot belong to more than $L$ task sets, the intersection of $|I|$ task sets is empty and hence, \eqref{eq:main} holds trivially. We henceforth assume $2 \leq |I| \leq L$. Suppose $n,n'\in I$, $n\neq n'$. Note that whenever there is a zero-waste transition from an \NT~to an \NMOT, each machine keeps its current task set and also takes $\DNNMO$ extra tasks. Hence, the intersection of a pair of task sets is increased by at most $2\DNNMO$ tasks. Therefore, 

\[
\begin{split}
&|\cap_{i \in I}S^{\Nm-r}_i| \leq |S^{\Nm-r}_n \cap S^{\Nm-r}_{n'}|\\
&\leq |\SNmn \cap \SNmnp| + \sum_{j = 0}^{r-1}2\Delta_{\Nm-j,\Nm-j-1}\\
&= |\SNmn \cap \SNmnp| + \sum_{j = 0}^{r-1}\bigg(\frac{2LF}{\Nm-j-1}-\frac{2LF}{\Nm-j}\bigg)\\
&\leq \frac{F}{\Nm}+\bigg(\frac{2LF}{\Nm-r}-\frac{2LF}{\Nm}\bigg)\\
&= \frac{F}{\Nm}+\frac{2LFr}{(\Nm-r)\Nm}. 
\end{split}
\]
Therefore, in order to show that \eqref{eq:main} holds for the $(\Nm-r,L,F)$-TAS $\S^{\Nm-r}$, that is, 
\[
|\cap_{i \in I}S^{\Nm-r}_i| \leq (\Nm-r-|I|)\Delta_{\Nm-r,\Nm-r-1},
\]
as we assume $|I| \leq L$, it suffices to show that
\[
\frac{F}{\Nm}+\frac{2LFr}{(\Nm-r)\Nm} \leq (\Nm-r-L)\Delta_{\Nm-r,\Nm-r-1},
\] 
or equivalently,
\begin{multline} 
\label{eq:bound}
\frac{F}{\Nm}+\frac{2LFr}{(\Nm-r)\Nm} \\ 
\leq (\Nm-r-L)\frac{LF}{(\Nm-r)(\Nm-r-1)}.
\end{multline} 
Simplifying \eqref{eq:bound}, we obtain 
\begin{multline} 
\label{eq:quadratic}
(2L-1)r^2-(3L\Nm-2\Nm-2L+1)r \\
+ \Nm(L\Nm-\Nm-L^2+1) \geq 0.
\end{multline} 
The left-hand side of \eqref{eq:quadratic} can be regarded as a quadratic polynomial in $r$, which has two positive roots
\[
\frac{(3L\Nm-2\Nm-2L+1)\pm\sqrt{\Delta}}{4L-2},
\]  
where $\Delta$ is given as in \eqref{eq:delta}. This is because when $L \geq 1$ and $\Nm \geq 2$, we have 
\[
\Delta = L\Nm\big((L\Nm-2) + 8(L-1)^2\big) + (2L-1)^2 > 0,
\]
and also, the coefficient of $r^2$ is $2L-1>0$, the coefficient of $r$ is negative, and the free coefficient is non-negative:
\begin{multline*}
\Nm(L\Nm-\Nm-L^2+1)\\ = \Nm(L-1)(\Nm-L-1) \geq 0.
\end{multline*}
Therefore, $R = 1 + \lfloor r_1\rfloor \geq 1$, where $r_1$ is the smaller (positive) root of \eqref{eq:quadratic}. Moreover, when 
\[
r \leq R-1 = \left\lfloor\frac{(3L\Nm-2\Nm-2L+1)-\sqrt{\Delta}}{4L-2}\right\rfloor,
\]
the left-hand side of \eqref{eq:quadratic} is non-negative, which implies that this inequality holds. Therefore, we have shown that for every $r < R$ defined as in \eqref{eq:R}, the inequality \eqref{eq:main} holds for the $(\Nm-r,L,F)$-TAS in consideration. Hence, there is a zero-waste transition from this TAS to an $(\Nm-r-1,L,F)$-TAS. Thus, $[\Nm-R,\Nm]$ is an $(L,F)$-zero-waste range.  
\end{proof} 

Equipped with Theorem~\ref{thm:ZWR}, we now present a few explicit zero-waste ranges based on known results on configurations from the literature of combinatorial designs. 

\begin{corollary} 
\label{cr:ZWR}
The following zero-waste ranges exist for all relevant $F$, that is, $F$ is divisible by $N(N-1)$ for every $N \in [\Nmi+1,\Nm]$.
\begin{enumerate}
	\item $L = 3$, $\Nm \geq 7$, $\Nmi = \Nm-\left\lfloor \frac{7\Nm-5-\sqrt{\Delta}}{10}\right\rfloor-1$, where $\Delta = 9\Nm^2+90\Nm+25$.\vspace{3pt}
	\item $L = 4$, $\Nm \geq 13$, $\Nmi = \Nm- \left\lfloor \frac{10\Nm-7-\sqrt{\Delta}}{14}\right\rfloor-1$, where $\Delta = 16\Nm^2+280\Nm+49$.\vspace{3pt}
	\item $L = q+1$, $\Nm = q^2+q+1$, $\Nmi = \Nm-\left\lfloor \frac{3q^3+4q^2+2q-\sqrt{\Delta}}{4q+2}\right\rfloor-1$, where $\Delta = q^6+12q^5+16q^4+4q^3-8q^2-12q-4$, for every prime power $q$.\vspace{3pt}
	\item $L = q$, $\Nm = q^2$, $\Nmi = \Nm-\left\lfloor \frac{3q^3-2q^2-2q+1-\sqrt{\Delta}}{4q-2}\right\rfloor-1$, where $\Delta = q^6+8q^5-24q^4+10q^3+4q^2-4q+1$, for every prime power $q$.\vspace{3pt}
	\item $L = q$, $\Nm = q^2-1$, $\Nmi = \Nm-\left\lfloor \frac{3q^3-2q^2-5q+3-\sqrt{\Delta}}{4q-2}\right\rfloor-1$, where $\Delta = q^6+8q^5-26q^4+2q^3+29q^2-14q+1$, for every prime power $q$.	
\end{enumerate}
\end{corollary} 
\begin{proof} 
Note that $(v,k)$-configurations exist for the following $v$ and $k$.
\begin{enumerate}
	\item $k \in \{3,4\}$ and $v \geq k(k-1)+1$ (See~\cite{colbourn2006handbook}).
	\item $k = q+1$ and $v = q^2+q+1$ for any prime power $q$. Such a $(q^2+q+1,q+1)$-configuration is also referred to as a finite projective plane. This gives us the Fano plane when $q = 2$. For this existence result and the following ones, see, e.g., \cite[p. 2]{Funk_etal_2009}. 
	\item $k = q$ and $v = q^2$ for any prime power $q$. A $(q^2,q)$-configuration can be obtained from a $(q^2+q+1,q+1)$-configuration by removing a point $P$ and all $q+1$ lines containing $P$ \emph{without} removing their points, and also removing one line \emph{containing} $P$ together with all of its points.
	\item $k = q$ and $v = q^2-1$ for any prime power $q$. A $(q^2-1,q)$-configuration can be obtained from a $(q^2+q+1,q+1)$-configuration by removing a point $P$ and all $q+1$ lines containing $P$ \emph{without} removing their points, and also removing one line \emph{not} containing $P$ together with all of its points.		 
\end{enumerate}
Applying Theorem~\ref{thm:ZWR} to these configurations, setting $\Nm = v$ and $L = k$, we deduce the conclusions of the corollary.
\end{proof} 

Applying Corollary~\ref{cr:ZWR} to the case $L = 3$ and $\Nm = 7$, we obtain a $(3,F)$-ZWR $[\Nmi=5,\Nm=7]$ where the $(7,3,F)$-TAS corresponds to the Fano plane. In other words, zero-waste transitions are possible between \emph{five} and \emph{seven} machines when $L = 3$.  Similarly, when applying the corollary to the case $L = 4$ and $\Nm = 13$, we obtain a $(4,F)$-ZWR \hoang{$[\Nmi=9,\Nm=13]$}, which implies that zero-waste transitions are possible between \emph{nine} and \emph{thirteen} machines. When $L=q$ and $\Nm = q^2$, for instance, we obtain a $(q,F)$-ZWR $\Nmm$ where $\Nmi = \Theta(\Nm/2)$. 
\hoang{Ideally, we would like to expand these ranges to $[\Nmi=L,\Nm]$ for every $\Nm > L$, which remains an open question.} 

\section{Experiments and Evaluations}
\label{sec:evaluations}

As discussed in Section~\ref{subsec:costETAS}, the case of machines joining seem less practical due to the extra communication overhead associated with data downloading. Therefore, we focus on the case of (one) machine leaving. 
We first performed simulations of different task allocation schemes in Python to evaluate the impact of the transition wastes on the \emph{CPU usage} and the \emph{computation time} for different sets of parameters (Section~\ref{subsec:sim}). 
We also implemented and ran these schemes on virtual machines for a specific set of parameters (corresponding to the Fano plane) to see the impact of the transition wastes on the actual \textit{completion time} of different schemes (Section~\ref{subsec:imp}).
These experiments demonstrate reasonable reductions in the CPU usage, the computation time, as well as the completion time, when shifted cyclic TAS or zero-weight TAS are used compared to the original cyclic TAS~\cite{Yang_etal_2019}. 

\subsection{Performance metrics}
\label{subsec:metrics}

\emph{First}, we note that from its definition (see Definition~\ref{def:TW}), the transition waste incurred at Machine $n$ when Machine $n^*$ leaves, i.e., $W_{n^*}(\SNn \to \SNMOn) \define \left|\SNn \Delta \SNMOn\right|-\frac{LF}{N(N-1)}$, is equal to two times the number of tasks \emph{abandoned} by Machine $n$, defined by $A_{n^*}(\SNn \to \SNMOn) \define \left|\SNn \setminus \SNMOn\right|$. We henceforth use \emph{abandoned} and \emph{wasted} interchangeably. \emph{Second}, in reality, the quantity $A_{n^*}(\SNn \to \SNMOn)$ only serves as an upper bound on the actual number of tasks abandoned by Machine $n$; the reason is that only those tasks already \emph{completed} by Machine $n$ when Machine $n^*$ left can be wasted. Tasks that were originally allocated to Machine $n$ but hadn't been executed by the time Machine $n^*$ left do not contribute to the (actual) transition waste\footnote{Investigating the actual transition waste given the list of completed tasks at all machines is a more general problem and left for future research.}. 
We ignored the coding/decoding time as this is the same for all schemes.

As such, to reflect the system's performance more accurately, we use $C_{n^*}(\SNn \to \SNMOn)$ to denote the set of \emph{completed} tasks (indices) at Machine $n$ when Machine $n^*$ left and observe that $A^{\text{comp}}_{n^*}\big(\SNn \to \SNMOn\big) \define \left|C_{n^*}(\SNn \to \SNMOn) \setminus  \SNMOn\right|$ is the number of tasks completed but abandoned (wasted) during the transition. We then use the following three different metrics to evaluate different task allocation schemes: the first metric represents the average waste in \emph{CPU usage} while the second and the third represent the impact of transition waste on the actual \emph{computation time} in slightly different ways. Regarding the CPU usage, as long as a task was executed and completed but not used, the CPU time spent on the task is consider wasted. 
The \textit{completion time} is the time the system requires from the start of computation until the heaviest loaded machine (the bottleneck) finishes all tasks\footnote{To avoid overcomplicating the discussion, we do not consider in our evaluation the straggler-tolerance capability of the underlying coded computing schemes. The analysis can be readily extended to that context by considering, for example, the second or higher-order maximum instead of the maximum.}. 
This is precisely the machine that wasted the largest number of completed tasks in the transition. 
That is why we need to examine the \emph{maximum} number of wasted completed tasks (over all active machines), which directly translates into the extra amount of time required for the system to complete the computation compared to the case when none of the completed tasks are wasted (as in a zero-waste TAS). In the following metrics, $\texttt{avg}$ is the abbreviation of ``average''.
\begin{itemize}
	\item $\aaw \define$\\ $\frac{1}{N}\sum_{n^* \in [N]}\bigg(\frac{1}{N-1}\sum_{n\in [N]\setminus \{n^*\}}A^{\text{comp}}_{n^*}\big(\SNn \to \SNMOn\big)\bigg)$, which is the average over all possible indices $n^* \in [N]$ of the average numbers of abandoned completed tasks $A^{\text{comp}}_{n^*}(\SNn \to \SNMOn)$ over all active machines $n \in [N]$, $n \neq n^*$. The higher $\aaw$, the higher waste in CPU usage \emph{on average}.
	\item $\amw \define$\\ $\frac{1}{N}\sum_{n^* \in [N]}\bigg(\max_{n\in [N]\setminus \{n^*\}}A^{\text{comp}}_{n^*}\big(\SNn \to \SNMOn\big)\bigg)$, which is the average over all possible indices $n^* \in [N]$ of the maximum numbers of wasted completed tasks $A^{\text{comp}}_{n^*}(\SNn \to \SNMOn)$ among all active machines $n \in [N]$, $n \neq n^*$. The higher $\amw$, the longer the \emph{averaged} computation time over all $n^*$.
	\item $\mmw \define$\\ $\max_{n^* \in [N]}\bigg(\max_{n\in [N]\setminus \{n^*\}}A^{\text{comp}}_{n^*}\big(\SNn \to \SNMOn\big)\bigg)$, which is the maximum among all possible indices $n^* \in [N]$ of the maximum numbers of wasted completed tasks $A^{\text{comp}}_{n^*}(\SNn \to \SNMOn)$ among all active machines $n \in [N]$, $n \neq n^*$. The higher $\mmw$, the longer the \emph{maximum} computation time among all $n^*$. 	
\end{itemize}
 
Note that for a zero-waste TAS, $C_{n^*}(\SNn \to \SNMOn) \subseteq \SNn \subseteq \SNMOn$, which implies that $A^{\text{comp}}_{n^*}\big(\SNn \to \SNMOn\big) = 0$. Alternatively, this can be deduced from the fact that the number of abandoned \emph{completed} tasks is not greater than the number of abandoned tasks, or half of the transition waste, which is zero in this case. Therefore, all the three metrics defined above are zero for a zero-waste TAS, which is the best possible. It remains to evaluate the performance of the cyclic and shifted cyclic schemes (against the zero-waste schemes).  
 
Here, we examine the transitions when a $\fr$ of the original tasks assigned to each machine have been completed, for $\fr \in \{0.1, 0.5, 0.9\}$. Note that $\left|C_{n^*}(\SNn \to \SNMOn)\right|$, i.e., the number of \emph{completed} tasks  at Machine $n$ is around $\fr\times \frac{LF}{N}$. As both cyclic and shifted cyclic TAS allocate a (cyclically) contiguous chunk of task indices to each machine, naturally, we assume that (cyclically) consecutive tasks are executed starting from the starting point of that set. We measure the percentage of the completed tasks that have been wasted and the percentage of the maximum number of the wasted completed tasks among all machines over the total number of tasks allocated to one machine when each machine has performed a fraction of $1/10$, $1/2$, and $9/10$ originally allocated tasks. Note that for each machine, the amount of \emph{extra} tasks it has to do compared to the case of zero waste ($LF/(N-1)$ tasks) is precisely the number of wasted completed tasks. Instead of using the three aforementioned metrics $\aaw$, $\amw$, and $\mmw$ directly, we transform them into percentages as follows.
\begin{itemize}
	\item $\aawp \define$\\ $100\times\aaw / \big(\fr \times (LF/N)\big)$: the percentage of the completed tasks that have been wasted (averaged over all active machines $n \neq n^*$ and then averaged over all $n^*\in [N]$).
	\item $\amwp \define$\\ $100\times\amw / \big(LF/(N-1)\big)$: the percentage of the wasted completed tasks over the total number of allocated task per machine (maximized over active machines $n \neq n^*$ and then averaged over $n^*\in [N]$).
	\item $\mmwp \define$\\ $100\times\mmw / \big(LF/(N-1)\big)$: the percentage of the wasted completed tasks over the total number of allocated task per machine (maximized over all active machines $n \neq n^*$ and over $n^*\in [N]$).
\end{itemize}  

We assume that each task takes the same amount of time to carry out (which makes sense because tasks correspond to computations over data of the same dimensions) and that machines have homogeneous computational capacities/loads (we focus on the performance evaluation of tasks allocation schemes and separate it from the underlying coded computing schemes, which consider stragglers). 
Then, the CPU usage and computation time of each TAS, as discussed earlier, can be captured accurately by the metrics defined in this section.

\begin{figure*}[!htb]
\centering
	\subfloat[Percentage of completed tasks that were wasted due to the transition (averaging over all active machines $n$ and over all machine $n^*$ that left). This represents the waste in CPU usage.]
     {\includegraphics[width=0.9\columnwidth]{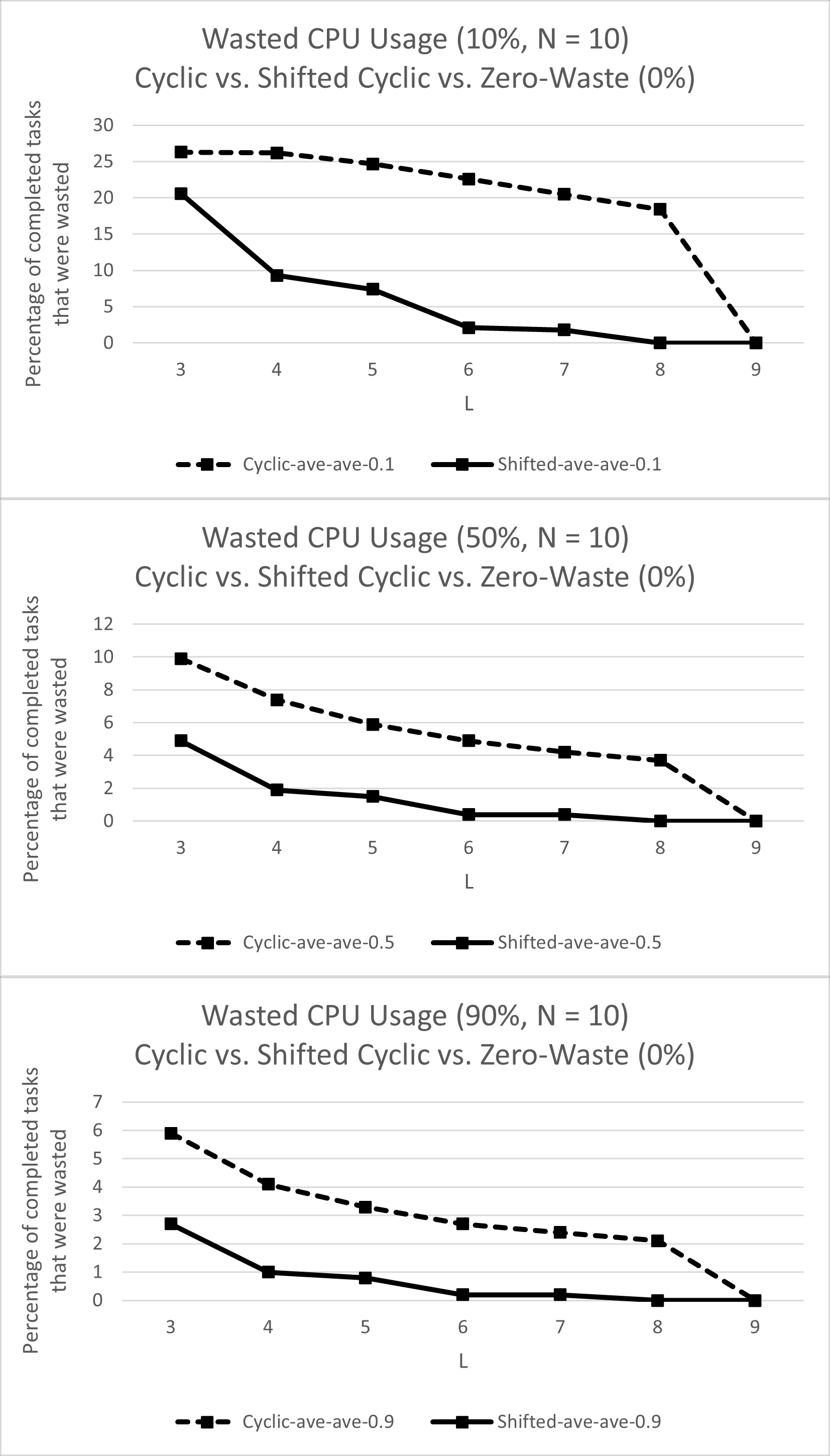}
     }
     \qquad
     \subfloat[The percentage of extra tasks (compared to a zero-waste scheme) due to wasted completed tasks in the transition. This corresponds to the overhead in computation time.]
     {\includegraphics[width=0.9\columnwidth]{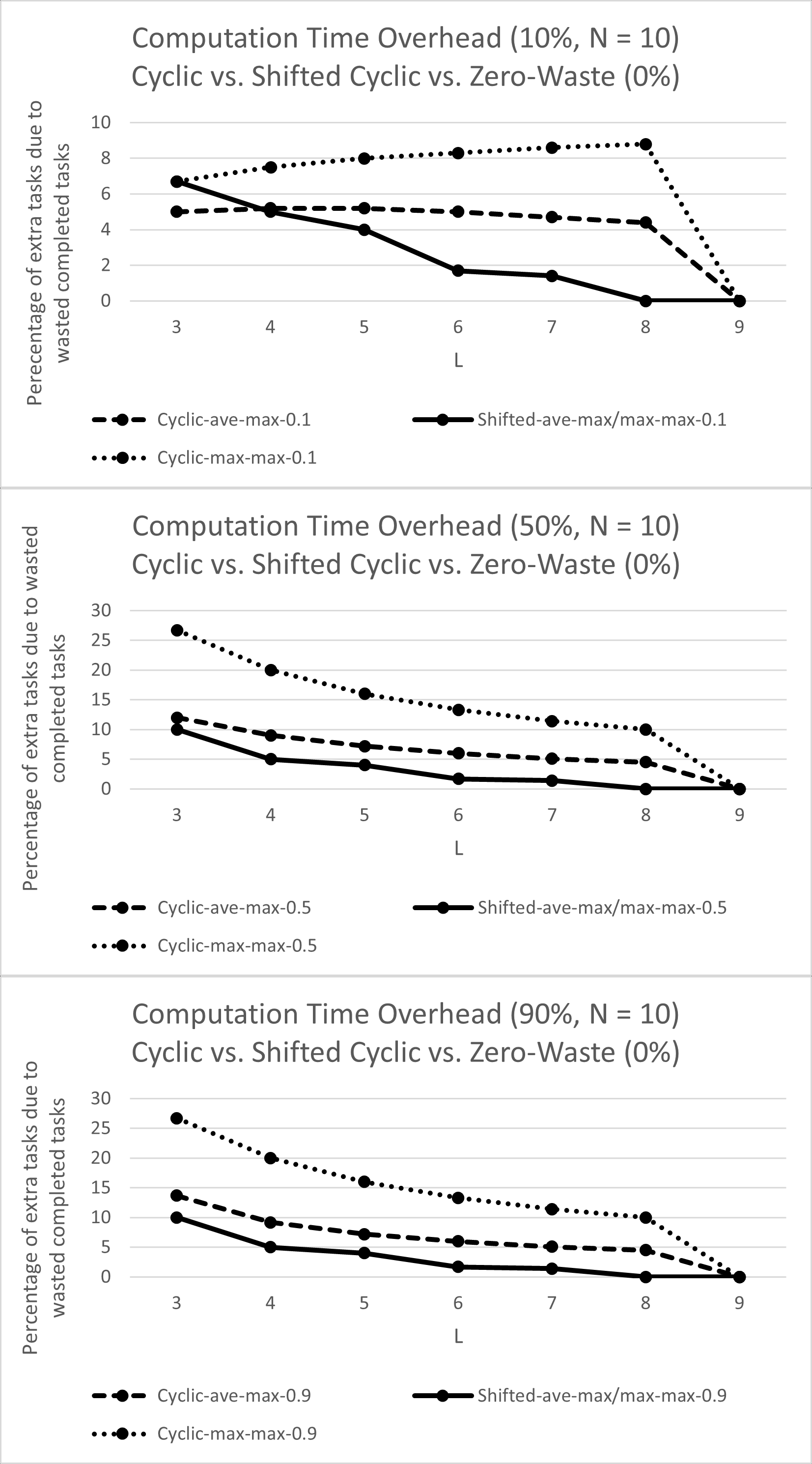}
     }
     \caption{The evaluation of the waste in CPU usage (measured by $\aawp$) and the overhead in computation time (measured by $\amwp$ and $\mmwp$) due to completed tasks being wasted during a transition when one machine leaves. We set up the transitions when 10\%, 50\%, and 90\% of the tasks originally allocated to one machine ($LF/N$ tasks) had been completed. We set $N = 10$ and $L \in \{3,4,\ldots,9\}$. We observe that the shifted cyclic TAS almost always performs better than the cyclic counterpart (except for a peculiar case when $L=3$ and $N=10$ at 10\%) and moreover, the gap between the performance of the cyclic and the zero-waste TAS decreases as $L$ grows.}
    \label{fig:N10}
\end{figure*}

\begin{figure*}[!htb]
\centering
	\subfloat[Percentage of completed tasks that were wasted due to the transition (averaging over all active machines $n$ and over all machine $n^*$ that left). This represents the waste in CPU usage.]
     {\includegraphics[width=0.9\columnwidth]{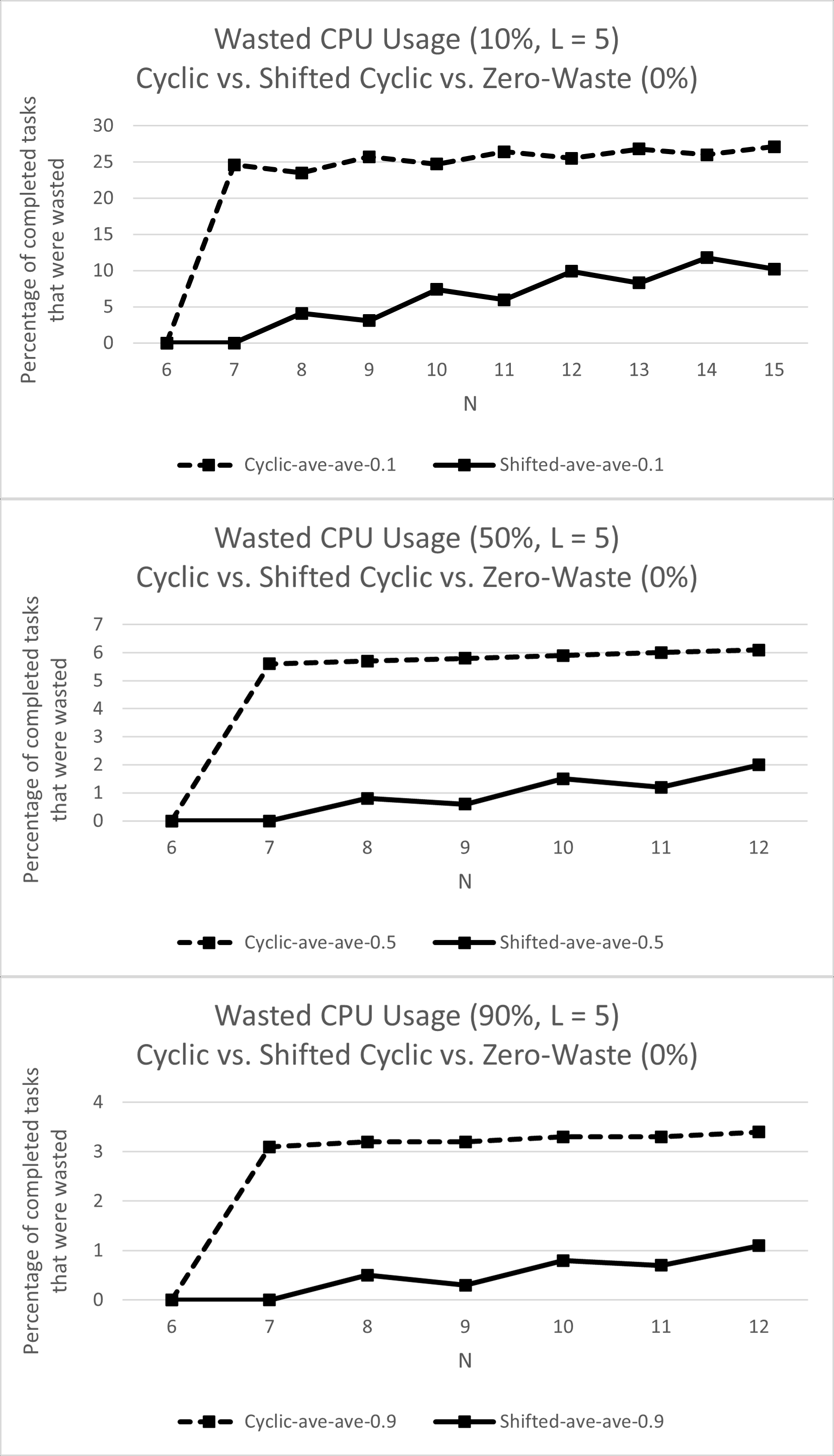}
     }
     \qquad
     \subfloat[The percentage of extra tasks (compared to a zero-waste scheme) due to wasted completed tasks in the transition. This corresponds to the overhead in computation time.]
     {\includegraphics[width=0.9\columnwidth]{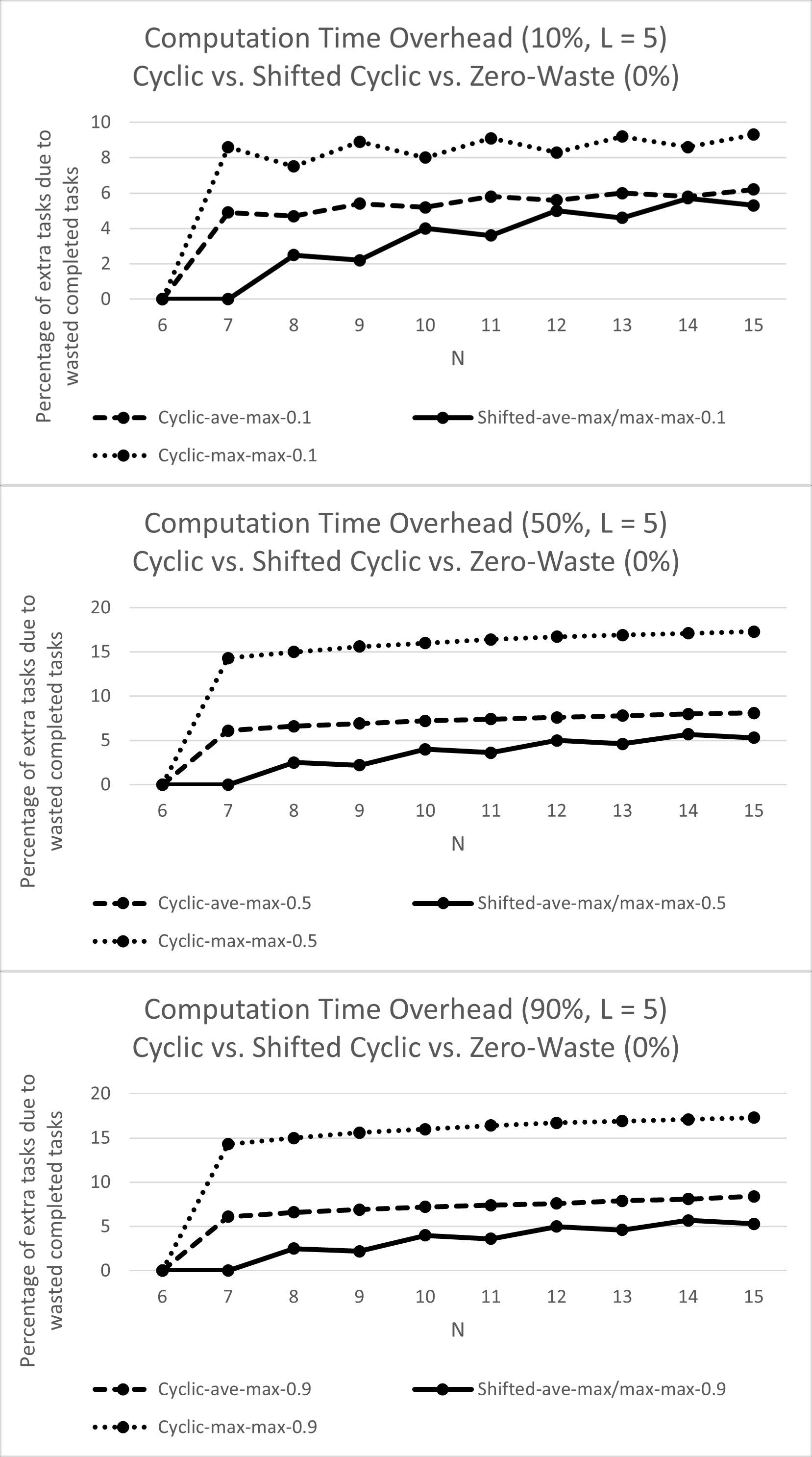}
     }
     \caption{The evaluation of the waste in CPU usage and the overhead in computation time due to wasted completed tasks in the cyclic and shifted cyclic schemes. Same as Fig.~\ref{fig:N10} but we set $L = 5$ and $N \in \{6,7,\ldots,15\}$ instead. We observe that the shifted cyclic TAS almost always performs better than the cyclic counterpart and moreover, the gap between the performance of the cyclic and the zero-waste TAS tends to increase as $N$ grows.}
    \label{fig:L5}
\end{figure*}

\subsection{Simulation}
\label{subsec:sim}

Our simulation results are summarized in Fig.~\ref{fig:N10}, where we fix $N = 10$ and let $L\in [3,4,\ldots,9]$, and in Fig.~\ref{fig:L5}, in which we fix $L = 5$ and let $N\in [6,7,\ldots,15]$. We let the transition (one machine leaving) happen when 10\%, 50\%, or 90\% of the originally allocated tasks to each machine had been completed. These are the points where there are some differences in the number of wasted completed tasks, which represent better the differences in the performance metrics among different schemes (e.g., compared to 25\%, 50\%, and 75\%). 
Note that all metrics used in Section~\ref{subsec:metrics} depend on the number of completed tasks that are wasted/abandoned. Our codes are available online at~\cite{python_sim}.

In summary, the shifted cyclic TAS almost always incurs less waste in CPU and smaller overhead in computation time compared to the cyclic TAS. We also observe that the gap between the performance of the cyclic/shifted cyclic TAS and the zero-waste TAS (no waste in CPU usage or computation time) grows \textit{gradually} when $N$ increases but shrinks more \textit{sharply} when $L$ increases. This is consistent with the derived formulas of the transition waste (in Theorem~\ref{thm:leaving} and Theorem~\ref{thm:shifted_leaving}), which serves as the upper bound on twice of the number of actual wasted completed tasks. Intuitively, this could also be explained by the fact that the total number of tasks allocated to each machine (the denominator of the metrics), that is $\frac{LF}{N}$, grows linearly with $L$, while the number of wasted completed tasks (the numerator of the metrics) seems to stay mostly independent of $L$.  
Thus, another quick take-away from the simulation is that the zero-waste TAS offers the largest gain over the cyclic TAS for small $L$ (minimum number of machines required by the system) and large $N$ (the number of available machines), and that $L$ plays a more significant role than $N$ in their performance.

\subsection{Implementation}
\label{subsec:imp}

We implemented the three task allocation schemes (cyclic, shifted cyclic, and zero-waste) on virtual machines and evaluated their performance (completion time) when $N = 7$, $L = 3$, and $F = 210$. The goal is to see if the completion times of different TAS are consistent with our simulation results. The selection of $N$ and $L$ is due to the parameter of the Fano plane (seven lines with three points per line). In theory, we only need $F=42$ to ensure that $F$ is divisible by $N=7$ and $N-1=6$, considering one machine leaving in our experiment. However, we set $F=210$, which is a medium number of tasks to make sure that the computation time is not too short to be ignored and at the same time, to avoid large overhead for the zero-waste scheme. Each task was carried out by multiplying a $2000\times 5000$ matrix and a vector of length $5000$ with integer entries randomly generated between -100 and 100. It took approximately 0.8 second to run each task at a worker.  
Each machine was initially allocated $LF/N = 90$ tasks when there are seven machines, and later with $LF/(N-1)=105$ tasks when one machine leaves. 

We used one virtual machine (the master) to run a bash script that coordinates the experiment on seven other virtual machines (the workers), all of which are Oracle cloud’s virtual machines VM.Standard.E2.1 with one OCPU, 0.7Gbs network bandwidth and 8GB of memory. The Python modules that performed the tasks were loaded into the workers. 
The master used \texttt{parallel-ssh} to send/retrieve data and commands to/from the workers.
First, the master set \texttt{time\_start} to be the start time and issue a command to run the Python modules on all seven workers.
To simulate the transition when one machine leaves at different times, we let the main Python module in each worker stop itself once it had completed 10\% (9 tasks), 50\% (45 tasks), and 90\% (81 tasks) of its originally allocated tasks (90 tasks), respectively. Each machine wrote into its log the list of tasks that had been completed. As all workers have the same configurations, they finished almost at the same time.

Once the master gathered that all workers had stopped, it removed one worker (Machine $n^*$) and issued another command to run the main Python modules on the six remaining ones (Machines $n=1,2,\ldots,7$ with $n \neq n^*$). At each remaining worker, the main Python module allocated a new set of tasks to the machine, depending on its index $n$ and the index of the machine that left $n^*$, and also on the particular task allocation scheme selected for that experiment (cyclic, shifted cyclic, or zero-waste). The list of tasks completed before the transition was read from its log and ignored because there is no need to run them the second time. Only tasks that hadn't been completed before were run.
The master then waited for all six workers to complete their allocated tasks and set \texttt{time\_end} to be the ending time.
The \textit{completion time} of the system was set to be $\texttt{completion\_time}=\texttt{time\_end}-\texttt{time\_start}$. 
This effectively recorded the \textit{maximum} running time among all remaining machines, which was then averaged out over all $n^*=1,2,\ldots,7$. 


\begin{figure*}[!htb]
\centering
	\subfloat[(Simulated) Computation overheads of the cyclic and shifted cyclic schemes compared to the zero-waste scheme (set at 0\%).]
     {\includegraphics[width=0.9\columnwidth]{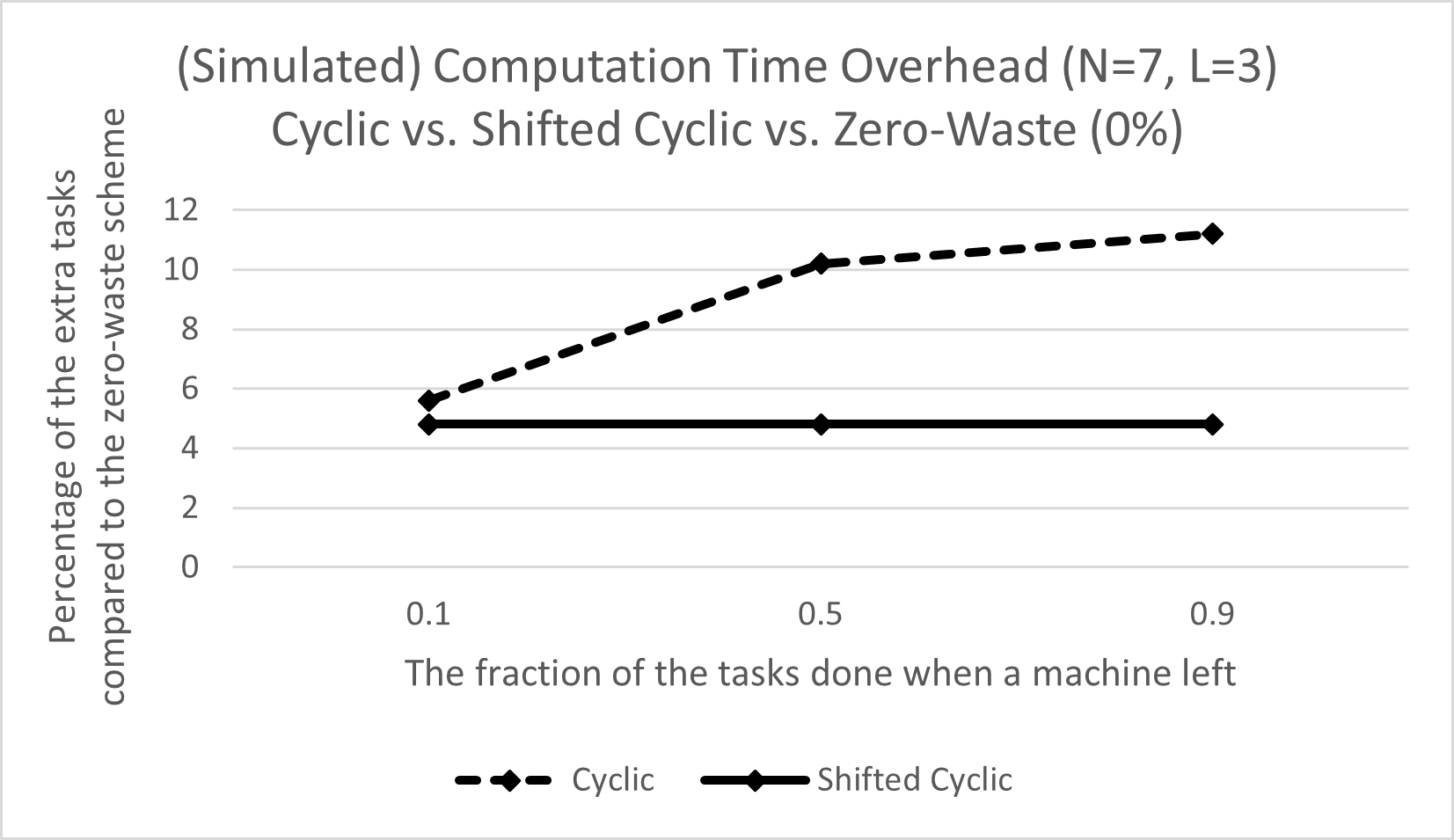}
     }
     \qquad
     \subfloat[Completion time overheads of the cyclic and shifted cyclic schemes compared to the zero-waste scheme (set at 0\%).]
     {\includegraphics[width=0.9\columnwidth]{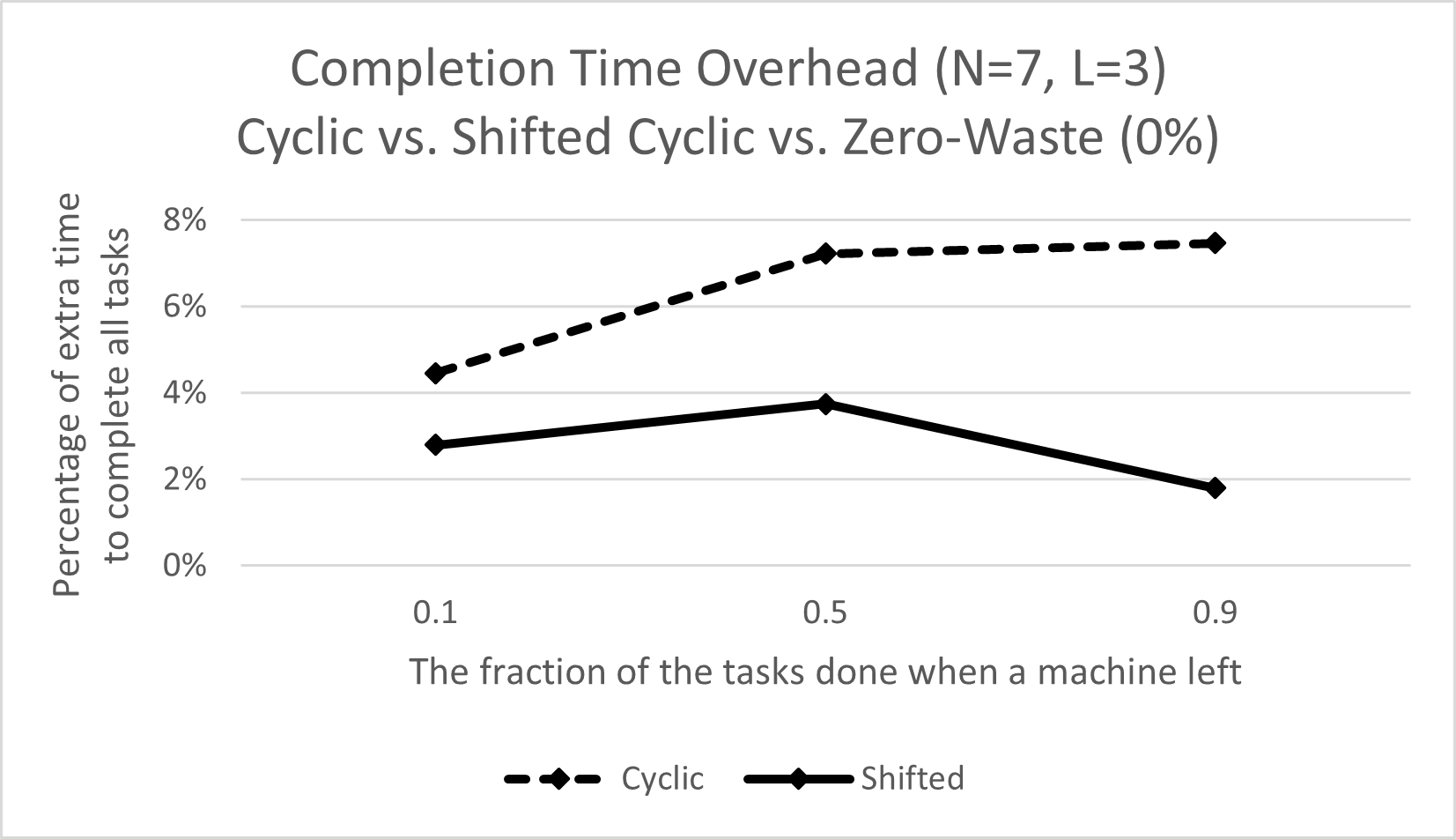}
     }
     \caption{The predicted computation time and the actual completion time overheads of the cyclic and shifted cyclic schemes versus the zero-waste scheme when $N = 7$, $L = 3$, and $F=210$, running on Oracle's virtual machines.}
    \label{fig:implement_N7_L3}
\end{figure*}

Note that the completion time includes the computation time and others such as I/O and communication overhead.
Compared to the simulated computation times (Fig.~\ref{fig:implement_N7_L3}(a)), the gaps in the completion times of these three schemes are smaller (Fig.~\ref{fig:implement_N7_L3}(b)). This was partly due to the impact of the communication overhead caused by \texttt{parallel-ssh} and of the I/O time (reading the large matrix into the memory). In total, the overhead, apart from computing the tasks, was approximately 20 seconds.
Better management of the communication and I/O may increase the impact of the computation time on the overall completion time.


\section{Conclusions}
\label{sec:conclusions}

Building up on the work of Yang \et{}~\cite{Yang_etal_2019} on coded elastic computing, we first propose a complete \emph{separation} between the elastic task allocation scheme and the coded computing scheme. As a result, we have the freedom to design \emph{efficient} elastic task allocation schemes as a combinatorial object \emph{independent} of the underlying coded computing schemes. Moreover, our result can be applied to almost every coded computing scheme developed in the literature. We illustrate the application of our result in matrix-vector and matrix-matrix multiplication, linear regression, and multivariate polynomial evaluation. 
The proposed separation \emph{simplifies} the coupling significantly compared to the original approach in~\cite{Yang_etal_2019}.  

Our main contributions in this work include the introduction of a new performance criterion for elastic task allocation schemes called the \emph{transition waste} and constructions of different schemes that achieve \emph{optimal} transition wastes. This quantity measures the number of tasks that available machines must abandon or take anew when one machine leaves or joins in the middle of the computation of a large scaled job. Smaller transition wastes reduce the waste of computing resources and speed up the job completion time. 

Our work and a few others~\cite{Yang_etal_2019,WoolseyChenJi-ISIT-2020, WoolseyChenJi-TCOM-2021, WoolseyKliewerChenJi-Globecom-2021,KianiAdikariDraper-ICASSP-2021} address the need to bridge the gap between the common setup of most coded computing schemes in the literature, where the number of available machines remain fixed, and an emerging trend in the cloud computing industry where the number of available machines can vary, due to the fact that low-priority virtual machines are often offered at much cheaper prices but can be taken back under a short notice (e.g. Amazon EC2 Spot and Microsoft Azure Batch). 

We can imagine one application of the coded elastic computing scheme as follows. We purchase a number of EC2 on-demand instances at a higher price while also get a few Spot instances at a much cheaper cost to run our computation. During the computation cycle, the low-priority Spot instances may leave, reducing the number of available machines. Our system can still handle this if we employ a coded elastic computing scheme in which the number of on-demand instances is greater than or equal to the minimum number of available machines required by the scheme. Thus, instead of maintaining all the costly on-demand instances from the beginning to the end, this approach allows us to take advantage of low-cost Spot instances available to us while keeping the computation run smoothly even when machines leave.
An interesting related approach from Amazon in 2018 was implemented in a new feature called Amazon EC2 Fleet~\cite{AmazonFleet}, which allows users to specify the target capacity and the preferred EC2 instances while automatically performs mix-and-match to meet customers specifications at a lowest price.      
       



\section*{Acknowledgement}
This work was supported by the Australian Research Council via the Discovery Project under Grant DP200100731. The implementation was done on Oracle Cloud virtual machines sponsored by the Oracle for Research.
The work of Yu-Chih Huang was supported by the Ministry of Science and Technology (MOST), Taiwan, under Grant MOST 111-2221-E-A49-069-MY3.

We thank Yaoqing Yang for helpful discussions and Pham Ngoc Duy and Khang Vo for their assistance in the implementation. We also thank the anonymous referees for many constructive comments, which helped to greatly improve the paper.

\bibliographystyle{IEEEtran}
\bibliography{ETAS}

\begin{thebibliography}{10}
\providecommand{\url}[1]{#1}
\csname url@samestyle\endcsname
\providecommand{\newblock}{\relax}
\providecommand{\bibinfo}[2]{#2}
\providecommand{\BIBentrySTDinterwordspacing}{\spaceskip=0pt\relax}
\providecommand{\BIBentryALTinterwordstretchfactor}{4}
\providecommand{\BIBentryALTinterwordspacing}{\spaceskip=\fontdimen2\font plus
\BIBentryALTinterwordstretchfactor\fontdimen3\font minus
  \fontdimen4\font\relax}
\providecommand{\BIBforeignlanguage}[2]{{%
\expandafter\ifx\csname l@#1\endcsname\relax
\typeout{** WARNING: IEEEtran.bst: No hyphenation pattern has been}%
\typeout{** loaded for the language `#1'. Using the pattern for}%
\typeout{** the default language instead.}%
\else
\language=\csname l@#1\endcsname
\fi
#2}}
\providecommand{\BIBdecl}{\relax}
\BIBdecl

\bibitem{Zaharia_etal_2010}
M.~Zaharia, M.~Chowdhury, M.~J. Franklin, S.~Shenker, and I.~Stoica, ``Spark:
  {C}luster computing with working sets,'' in \emph{Proceedings of the 2nd
  USENIX Conference on Hot Topics in Cloud Computing}, 2010.

\bibitem{Hadoop}
Apache Hadoop. \url{http://hadoop.apache.org}.

\bibitem{Ananthanarayanan_etal_2013}
G.~Ananthanarayanan, A.~Ghodsi, S.~Shenker, and I.~Stoica, ``Effective
  straggler mitigation: {A}ttack of the clones,'' in \emph{Proceedings of the
  10th USENIX Conference on Networked Systems Design and Implementation}, 2013,
  pp. 185--198.

\bibitem{Dean_etal_2013}
F.~Dean and L.~A. Barroso, ``The tail at scale,'' \emph{Communications of the
  ACM}, vol.~56, no.~2, pp. 74--80, 2013.

\bibitem{Yadwadkar_etal_2016}
N.~J. Yadwadkar, B.~Hariharan, J.~E. Gonzalez, and R.~Katz, ``Multi-task
  learning for straggler avoiding predictive job scheduling,'' \emph{Journal of
  Machine Learning Research}, vol.~17, no.~1, pp. 3692--3728, 2016.

\bibitem{Lee_etal_2018}
K.~{Lee}, M.~{Lam}, R.~{Pedarsani}, D.~{Papailiopoulos}, and K.~{Ramchandran},
  ``Speeding up distributed machine learning using codes,'' \emph{IEEE
  Transactions on Information Theory}, vol.~64, no.~3, pp. 1514--1529, 2018.

\bibitem{Li_etal_2015}
S.~{Li}, M.~A. {Maddah-Ali}, and A.~S. {Avestimehr}, ``Coded {M}ap{R}educe,''
  in \emph{Proceedings of the 53rd Annual Allerton Conference on Communication,
  Control, and Computing}, 2015, pp. 964--971.

\bibitem{Tandon_etal_2017}
R.~Tandon, Q.~Lei, A.~G. Dimakis, and N.~Karampatziakis, ``Gradient coding:
  {A}voiding stragglers in distributed learning,'' in \emph{Proceedings of the
  34th International Conference on Machine Learning}, 2017, pp. 3368--3376.

\bibitem{Huang_etal_1984}
K.~H. Huang and J.~A. Abraham, ``Algorithm-based fault tolerance for matrix
  operations,'' \emph{IEEE Transactions on Computers}, vol. C-33, no.~6, pp.
  518--528, 1984.

\bibitem{Dutta_etal_2016}
S.~Dutta, V.~Cadambe, and P.~Grover, ``Short-dot: {C}omputing large linear
  transforms distributedly using coded short dot products,'' in \emph{Advances
  in Neural Information Processing Systems}, 2016, pp. 2100--2108.

\bibitem{Yu_etal_2017}
Q.~Yu, M.~A. Maddah-Ali, and S.~Avestimehr, ``Polynomial codes: an optimal
  design for high-dimensional coded matrix multiplication,'' in \emph{Advances
  in Neural Information Processing Systems}, 2017, pp. 4403--4413.

\bibitem{Karakus_etal_2017}
C.~{Karakus}, Y.~{Sun}, S.~{Diggavi}, and W.~Yin, ``Straggler mitigation in
  distributed optimization through data encoding,'' in \emph{Advances in Neural
  Information Processing Systems 30}, 2017, pp. 5434--5442.

\bibitem{Li_etal_2018}
S.~{Li}, M.~A. {Maddah-Ali}, Q.~{Yu}, and A.~S. {Avestimehr}, ``A fundamental
  tradeoff between computation and communication in distributed computing,''
  \emph{IEEE Transactions on Information Theory}, vol.~64, no.~1, pp. 109--128,
  2018.

\bibitem{Mallick_etal_2018}
A.~Mallick, M.~Chaudhari, U.~Sheth, G.~Palanikumar, and G.~Joshi, ``Rateless
  codes for near-perfect load balancing in distributed matrix-vector
  multiplication,'' \emph{Communications of the ACM}, vol.~65, no.~5, pp.
  111--118, 2022.

\bibitem{Yu_etal_2019}
Q.~Yu, S.~Li, N.~Raviv, S.~M.~M. Kalan, M.~Soltanolkotabi, and S.~Avestimehr,
  ``Lagrange coded computing: Optimal design for resiliency, security, and
  privacy,'' in \emph{Proceedings of Machine Learning Research}, vol.~89, 2019,
  pp. 1215--1225.

\bibitem{Kosaian_etal_2018}
\BIBentryALTinterwordspacing
J.~Kosaian, K.~V. Rashmi, and S.~Venkataraman, ``Learning a code: {M}achine
  learning for approximate non-linear coded computation,'' 2018. [Online].
  Available: \url{http://arxiv.org/abs/1806.01259}
\BIBentrySTDinterwordspacing

\bibitem{Kosaian_etal_2020}
------, ``Learning-based coded computation,'' \emph{IEEE Journal on Selected
  Areas in Information Theory}, vol.~1, no.~1, pp. 227--236, 2020.

\bibitem{AmazonSpot}
Amazon EC2 Spot Instances. \url{https://aws.amazon.com/ec2/spot/}.

\bibitem{AzureBatch}
Microsoft\hspace{-2pt} Azure\hspace{-2pt} Batch.\hspace{-2pt}
  \url{https://azure.microsoft.com/en-au/services/batch/}.

\bibitem{Yang_etal_2019}
Y.~Yang, P.~Grover, and S.~Kar, ``Coded elastic computing,'' in \emph{IEEE
  International Symposium on Information Theory}, 2019, pp. 2654--2658.

\bibitem{WoolseyChenJi-ISIT-2020}
N.~Woolsey, R.-R. Chen, and M.~Ji, ``Heterogeneous computation assignments in
  coded elastic computing,'' in \emph{IEEE International Symposium on
  Information Theory}, 2020, pp. 168--173.

\bibitem{WoolseyChenJi-TCOM-2021}
------, ``Coded elastic computing on machines with heterogeneous storage and
  computation speed,'' \emph{IEEE Transactions on Communications}, vol.~69,
  no.~5, pp. 2894--2908, 2021.

\bibitem{WoolseyKliewerChenJi-Globecom-2021}
N.~Woolsey, J.~Kliewer, R.-R. Chen, and M.~Ji, ``A practical algorithm design
  and evaluation for heterogeneous elastic computing with stragglers,'' in
  \emph{IEEE Global Communications Conference}, 2021, pp. 1--6.

\bibitem{KianiAdikariDraper-ICASSP-2021}
S.~Kiani, T.~Adikari, and S.~C. Draper, ``Hierarchical coded elastic
  computing,'' in \emph{IEEE International Conference on Acoustics, Speech and
  Signal Processing}, 2021, pp. 4045--4049.

\bibitem{MuntzLui1990}
R.~R. Muntz and J.~C.~S. Lui, ``Performance analysis of disk arrays under
  failure,'' in \emph{Proceedings of the 16th International Conference on Very
  Large Data Bases}, ser. VLDB '90, 1990, pp. 162--173.

\bibitem{HollandGibson1992}
M.~Holland and A.~G. Gibson, ``Parity declustering for continuous operation in
  redundant disk arrays,'' in \emph{Proceedings of the Fifth International
  Conference on Architectural Support for Programming Languages and Operating
  Systems}, ser. ASPLOS V, 1992, pp. 23--35.

\bibitem{DauJiaJinXiChan2014}
S.~H. {Dau}, Y.~{Jia}, C.~{Jin}, W.~{Xi}, and K.~S. {Chan}, ``Parity
  declustering for fault-tolerant storage systems via $t$-designs,'' in
  \emph{2014 IEEE International Conference on Big Data}, 2014, pp. 7--14.

\bibitem{Hall}
P.~Hall, ``On representatives of subsets,'' \emph{Journal of the London
  Mathematical Society}, vol. s1-10, no.~1, pp. 26--30, 1935.

\bibitem{AhujaMagnantiOrlin}
R.~K. Ahuja, R.~L. Magnanti, and J.~B. Orlin, \emph{Network Flows}.\hskip 1em
  plus 0.5em minus 0.4em\relax Englewood Cliffs, NJ: Prentice-Hall, 1993.

\bibitem{colbourn2006handbook}
C.~J. Colbourn and J.~H. Dinitz, \emph{Handbook of Combinatorial Designs,
  Second Edition (Discrete Mathematics and Its Applications)}.\hskip 1em plus
  0.5em minus 0.4em\relax CRC Press, 2006.

\bibitem{Funk_etal_2009}
M.~Funk, D.~Labbate, and V.~Napolitano, ``Tactical (de-)compositions of
  symmetric configurations,'' \emph{Discrete Mathematics}, vol. 309, no.~4, pp.
  741--747, 2009.

\bibitem{python_sim}
\BIBentryALTinterwordspacing
Python codes for task allocations in coded elastic computing. [Online].
  Available: \url{https://github.com/dausonhoang/coded_elastic_computing}
\BIBentrySTDinterwordspacing

\bibitem{AmazonFleet}
Introducing Amazon EC2 Fleet.
  \url{https://aws.amazon.com/about-aws/whats-new/2018/04/introducing-amazon-ec2-fleet/}.

\bibitem{Fahim2017}
M.~Fahim, H.~Jeong, F.~Haddadpour, S.~Dutta, V.~Cadambe, and P.~Grover, ``On
  the optimal recovery threshold of coded matrix multiplication,'' in
  \emph{Proceedings of the 55th Annual Allerton Conference on Communication,
  Control, and Computing}, 2017, pp. 1264--1270.

\end{thebibliography}

\section{Appendix}

\subsection{Coupling an Elastic Task Allocation Scheme and a Coded Computing Scheme}
\label{app:coupling}

\textbf{Matrix-Matrix Multiplication.}
The goal is to compute the product $\bA\bB$, where $\bA$ and $\bB$ are matrices of matching dimensions, in the presence of $E$ stragglers $(0 \leq E < L)$ and with a varied number of available machines $N$ $(L \leq N \leq \Nm)$.

We partition $\bA$ and $\bB$ column-wise and row-wise, respectively, into $F$ equal-sized sub-matrices (padding with zeros if necessary) as follows,   \hspace{-15pt}
\[
    \bA=\begin{bmatrix}
                 \bA_0, & \bA_1, & \ldots & \bA_{F-1} \\
               \end{bmatrix},\quad
               \bB=\begin{bmatrix}
                            \bB_0 \\
                            \bB_1 \\
                            \vdots \\
                            \bB_{F-1} \\
                          \end{bmatrix}.
\]
The pair $(\bA_f, \bB_f)$, $f\in[[F]]$, forms the $f$th sub-instance and the computation of $\bA_f\bB_f$ is referred to as Task~$f$. As $\bA\bB = \sum_{f=0}^{F-1}\bA_f\bB_f$, the completion of all $F$ tasks gives us the product $\bA\bB$. 
For each Task~$f$, an existing CCS for matrix-matrix multiplication can be applied (e.g., MatDot~\cite{Fahim2017}).

\textbf{Linear Regression.}
Given a data matrix $\bX$ and a vector $\by$, we aim to find a weight vector $\bw$ that minimizes the loss function $\|\bX\bw -y \|^2$. Using gradient descent, in each iteration, we update the weight using the gradient of the loss function, which requires the computation of $\bX^{\text{T}}(\bX\bw^{(t)}-\by)$.  

The algorithm in~\cite{Yang_etal_2019} first computes $\bX\bw^{(t)}$ via coded elastic computing, computes $\bz^{(t)}=\bX\bw^{(t)}-\by$ at the master node, and adaptively encodes $\bz^{(t)}$ according to the knowledge of machines that are active. Hence, it is not suitable for the scenario where machines join or leave in the middle of each iteration.
Our approach presented below simplifies the approach in~\cite{Yang_etal_2019} and also overcomes its drawback. 

Note that both $\bX$ and $\by$ are fixed while $\bw^{(t)}$ varies from one iteration to the next. 
Therefore, the matrix-matrix product $\bA=\bX^T\bX$ and the matrix-vector product $\bX^T\by$ can be computed once in advance with amortized cost using an ETAS as described earlier. The only job left is to repeatedly compute $\bA\bw^{(t)}$, $t = 0,1,\ldots$ Again, we use an ETAS to perform this matrix-vector multiplication. Despite of its conceptual simplicity, this procedure not only allows machines join or leave in the middle of each iteration but also saves communication bandwidth as at each iteration, we only send $\bw^{(t)}$ to machines rather than both $\bw^{(t)}$ and a coded version of $\bz^{(t)}$.

\textbf{Multivariate polynomial evaluation.}
We aim to compute $g(\bX_1),\ldots,g(\bX_K)$, where $g$ is a multivariate polynomial and $\bX_k$ is a large matrix or vector $(k \in [K])$, in a way that tolerates $E$ stragglers and allows the number of available machines vary between $L$ and $\Nm$. 

Suppose that $K$ is divisible by $F$ (padding if necessary). We partition the set of evaluation points into $F$ equal parts
\[
\P_f = \left\{\bX_{fK/F+1},\ldots,\bX_{fK/F+K/F}\right\},\quad f \in [[F]].
\] 
Task~$f$ refers to the computations of $g(\bX_p)$, $p \in \P_f$. Clearly, the completion of all $F$ tasks gives us $g(\bX_1),\ldots,g(\bX_K)$ as desired. 
Yu~\et{}~\cite{Yu_etal_2019} propose a CCS called the Lagrange coded computing to perform distributed polynomial evaluation that tolerates stragglers. We can apply this CCS to each task using $\Nm$ machines and recovery threshold $L-E$.

\subsection{Proof of Lemma~\ref{lem:leave_2}}
\label{app:proof_leave_2}

\begin{proof}[Proof of Lemma~\ref{lem:leave_2}]
As $n > n^*$, we have
\[
S^{N-1}_n = \bigg[(n-2)\frac{F}{N-1},(n-2)\frac{F}{N-1}+\frac{LF}{N-1}-1\bigg]\hspace{-10pt} \pmod F.
\]
We now apply Lemma~\ref{lem:sym} to the sets 
\[
S = \SNMOn = [a,b]\hspace{-5pt}\pmod F, \quad T = \SNn=[c,d]\hspace{-5pt}\pmod F.
\] 
The common assumptions of Lemma~\ref{lem:sym} are verified as follows. We have
\[
0 \leq a = (n-2)\frac{F}{N-1} < c = (n-1)\frac{F}{N} < F,\]
\[
0 < |S| = \frac{LF}{N-1} < F,\quad 0 < |T| = \frac{LF}{N} < F.
\]

\textbf{Case 1.} When $n^* \geq N-L$ or $n^* < N-L$ but $n > N-L$, we aim to show $W_{n^*}(\SNn \to \SNMOn) = 0$ by proving that $\SNn \subset \SNMOn$ (Lemma~\ref{lem:zero_trivial}). Note that in this case, we always have $n \geq N-L+1$. Therefore, 
\begin{equation}
\label{eq:inequality}
\frac{LF}{N(N-1)} \geq \frac{(N-n+1)F}{N(N-1)},  
\end{equation}
which is equivalent to $|\SNMOn| - |\SNn| \geq (c-a)$, or $|S| \leq (c-a)+T$. By Lemma~\ref{lem:sym}~(b), we conclude that $T = \SNn \subset S = \SNMOn$, as desired. Hence the transition waste incurred at Machine $n$ is zero. 

\textbf{Case 2.} Suppose that $n^* < n \leq N-L$. The inequality \eqref{eq:inequality} is reversed, which gives us $|S| < (c-a)+|T|$. We now verify that other conditions of Lemma~\ref{lem:sym}~(a) are also satisfied. First, it is clear that
\[
c-a = \frac{(N-n+1)F}{N(N-1)} < \frac{LF}{N-1} = |S|. 
\]
Moreover, as $N > L+1$ (our assumption),  
\[
(c-a) + |T| = \frac{(N-n+1)F}{N(N-1)} + \frac{LF}{N} < F. 
\] 
Therefore, by Lemma~\ref{lem:sym}~(a), we obtain 
\[
\begin{split}
|\SNn \Delta \SNMOn|
&= 2(c-a) + (|\SNn|-(|\SNMOn|)\\ 
&= \frac{2(N-n+1)F}{N(N-1)} - \frac{LF}{N(N-1)}.
\end{split}
\]
Noting that $\DNNMO = \frac{LF}{N(N-1)}$, we obtain
\[
\begin{split}
W_{n^*}(\SNn \to \SNMOn) &= |\SNn \Delta \SNMOn| - \DNNMO\\
&= \frac{2(N-L-n+1)F}{N(N-1)}.
\end{split}\vspace{-10pt}
\] 
This completes the proof.
\end{proof} 

\subsection{Proof of Theorem~\ref{thm:shifted_joining}}
\label{app:proof_shifted_joining}

\begin{proof}[Proof of Theorem~\ref{thm:shifted_joining}]
Without loss of generality, we can always assume that $\delta' = 0$ and $\delta = \hspace{-2pt} \lfloor \frac{N+L-1}{2} \rfloor \frac{F}{N(N+1)}$. We provide a proof when $N+L$ is odd, i.e., $\delta =\frac{(N+L-1)F}{2N(N+1)}$ noting that we assume $N(N+1)$ divides $F$ (padding with dummy tasks if necessary). A proof for the case when $N+L$ is even can be done similarly. 

With $\delta' = 0$ and $\delta =\frac{(N+L-1)F}{2N(N+1)}$, we have
\[
\SDPCN = (S^N_1,\ldots,S^N_N),\quad
\SDCNPO = (S^{N+1}_1,\ldots,S^{N+1}_{N+1}),
\]
where for $n \in [N]$, 
\[
\SNn = \left[(n-1)\frac{F}{N}, (n-1)\frac{F}{N} + \frac{LF}{N}-1\right]\hspace{-10pt} \pmod F.
\]
\begin{multline*}
S^{N+1}_n = \bigg[(n-1)\frac{F}{N+1}+\frac{(N+L-1)F}{2N(N+1)},\\ 
(n-1)\frac{F}{N+1} + \frac{LF}{N+1}-1+\frac{(N+L-1)F}{2N(N+1)}\bigg] \pmod F.
\end{multline*}
To compute the transition waste $W(\SNn \to S^{N+1}_n)$ incurred at Machine $n \in [N]$, we consider the following three cases.

\textbf{Case 1.} $1 \leq n < \frac{N-L+1}{2}$. It can be easily verified that all conditions of Lemma~\ref{lem:sym}~(a) are satisfied for $S \define \SNn = [a,b] \pmod F$ and $T \define S^{N+1}_n = [c,d] \pmod F$. Therefore, 
\[
\begin{split}
W(\SNn \to S^{N+1}_n) &= 2(c-a)-2\DNNPO\\
&= \frac{(N+L+1-2n)F}{N(N+1)} - \frac{2LF}{N(N+1)}\\
&= \frac{(N-L+1-2n)F}{N(N+1)}.
\end{split}
\]

\textbf{Case 2.} $\frac{N-L+1}{2} \leq n < \frac{N+L+1}{2}$. We can verify that all conditions of Lemma~\ref{lem:sym}~(b) are satisfied for $S \define \SNn = [a,b] \pmod F$ and $T \define S^{N+1}_n = [c,d] \pmod F$. Hence, $T \subset S$ and 
$W(\SNn \to S^{N+1}_n) = 0.$

\textbf{Case 3.} $\frac{N+L+1}{2} \leq n \leq N$. We can verify that all conditions of Lemma~\ref{lem:sym}~(a) are satisfied for $S \define S^{N+1}_n = [a,b] \pmod F$ and $T \define S^{N}_n = [c,d] \pmod F$. Therefore, 
\[
\begin{split}
W(\SNn \to S^{N+1}_n) &= 2(c-a) + \DNNPO - \DNNPO\\
&= \frac{(2n-(N+L+1))F}{N(N+1)}.
\end{split}
\]

Thus, the waste when transitioning from $\SDCN$ to $\SDPCNPO$ is

\begin{multline*}
W(\SDCN \to \SDPCNPO) =
\frac{F}{N(N+1)} \bigg( \sum_{n=1}^{\frac{N-L-1}{2}} (N-L+1-2n)\\ 
+\sum_{n=\frac{N-L+1}{2}}^{\frac{N+L-1}{2}}0
+ \sum_{n=\frac{N+L+1}{2}}^N (2n-(N+L+1))
\bigg)\\
= \frac{(N-L-1)(N-L+1)F}{2N(N+1)}.\vspace{-10pt}
\end{multline*}
This completes the proof. 
\end{proof} 

\subsection{Proof of Theorem 5}
\label{sub:AppendixOptProof}

Note that we only need to prove Theorem~\ref{thm:optimal} for the case when Machine $N+1$ joins. The following lemma holds for all $\delta \in [[F]]$.

\begin{lemma}
\label{lem:opt}
The transition waste when transitioning from a cyclic \NT~$\SCN$ to a $\delta$-shifted cyclic \NPOT~$\SDCNPO$ is 
\[
{W(\SCN\to \SDCNPO)} = \texttt{Sum1} + \texttt{Sum2} + \texttt{Sum3},
\]
where these three sums are given as follows.
Setting $d = \frac{F}{N(N+1)} \in \bbZ$, the first sum is 
\[
\texttt{Sum1} = \sum_{n \in [N]: (n-1)d > \delta} 2((n-1)d-\delta).
\]
When $L < \lceil \frac{N+1}{2} \rceil$, the second and third sums are 
\[
\begin{split}
\texttt{Sum2} &= 
\sum_{n \in [N]: (n-1+L)d \leq \delta < (n-1+L+LN)d}\hspace{-30pt} 2\big( \delta - (n-1+L)d \big),\\ 
\texttt{Sum3} &= 
\sum_{n \in [N]: (n-1+L+LN)d \leq \delta \leq F+(n-1)d-LNd}\hspace{-50pt} 2LNd\\
&\quad\qquad +\sum_{n \in [N]: F+(n-1)d-LNd < \delta}\hspace{-15pt} 2\big( F+(n-1)d-\delta \big).
\end{split}
\]
When $L \geq \lceil \frac{N+1}{2} \rceil$, the second and third sums are 
\[
\begin{split}
\texttt{Sum2} &= 
\sum_{n \in [N]: (n-1+L)d \leq \delta \leq F + (n-1)d-LNd}\hspace{-30pt} 2\big( \delta - (n-1+L)d \big)\\ 
&\ + \sum_{F+(n-1)d-LNd < \delta < (n-1+L+LN)d}\hspace{-30pt} 2(N-L)F/N,\\
\texttt{Sum3} &= 
\sum_{n \in [N]: (n-1+L+LN)d  \leq \delta}\hspace{-30pt} 2\big( F+(n-1)d-\delta \big).
\end{split}
\]
\end{lemma}
\begin{proof} 
These sums are obtained by considering all possible cases of the intersection between $\SNn$ and $\SNPOn$ taking into account the fact that we have shifted $\SNPOn$ cyclicly by $\delta$ positions compared to the ordinary cyclic TAS. 

Let $\SCN = (\SN_1,\ldots,\SN_N)$ and $\SDCNPO = (\SNPO_1,\ldots,\SNPO_{N+1})$.
For $n \in [N]$, 
\[
\SNn = \left[(n-1)\frac{F}{N}, (n-1)\frac{F}{N} + \frac{LF}{N}-1\right]\hspace{-10pt} \pmod F.
\]
\begin{multline*}
S^{N+1}_n = \bigg[(n-1)\frac{F}{N+1}+\delta,\\ 
(n-1)\frac{F}{N+1} + \frac{LF}{N+1}-1+\delta\bigg] \pmod F.
\end{multline*}
To compute the transition waste $W(\SNn \to S^{N+1}_n)$ incurred at Machine $n \in [N]$, we consider the following three cases depending on the relative position of the endpoints of $\SNn$ and $\SNPOn$ on the circle of integers mod $F$.

\textbf{Case 1.} $\delta < \frac{(n-1)F}{N(N+1)} = (n-1)d$. The left endpoint of $\SNPOn$ lies between $0$ and the left endpoint of $\SNn$ (see Fig.~\ref{fig:opt_case1}). Applying Lemma~\ref{lem:sym}~(a) to $S = S^{N+1}_n$ and $T = S^N_n$, we have
\[
W(\SNn \to S^{N+1}_n) = 2((n-1)d-\delta).
\]
Case~1 gives rise to \texttt{Sum1}.
\begin{figure}[!htb]
\centering
\includegraphics[scale=0.6]{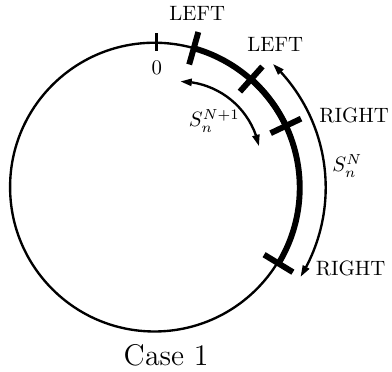}
\caption{Illustration of Case 1.}
\label{fig:opt_case1}
\end{figure}
\vspace{-10pt}
\begin{figure}[!htb]
\centering
\includegraphics[scale=0.6]{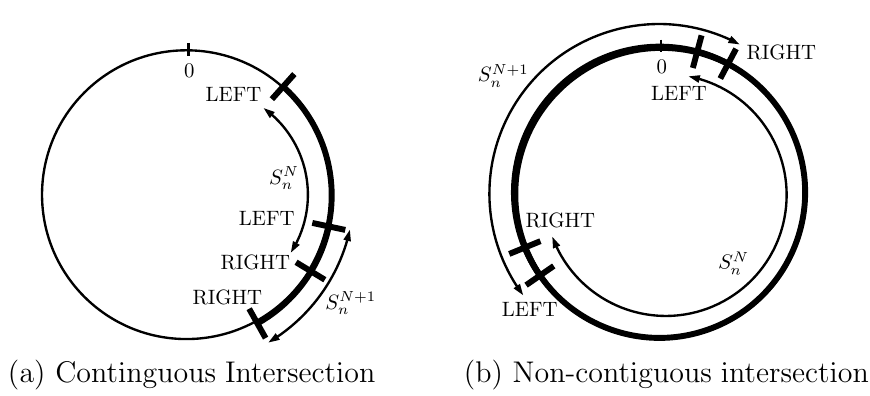}
\caption{Illustration of scenarios in Sub-case 2.2.}
\label{fig:opt_case2}
\end{figure}

\textbf{Case 2.} $(n-1)d \leq \delta < (n-1+L+LN)d$.
The left endpoint of $\SNPOn$ lies between the two endpoints of $\SNn$ (inclusive). We further divide Case~2 into two sub-cases. 

\emph{Sub-case 2.1.} $(n-1)d \leq \delta < (n-1+L)d$. Since $\SNPOn \subset \SNn$, by Lemma~\ref{lem:zero_trivial}, the transition waste is zero and we can ignore this sub-case.  

\emph{Sub-case 2.2.} $(n-1+L)d \leq \delta < (n-1+L+LN)d$. When $L < \lceil \frac{N+1}{2} \rceil$, the intersection of $\SNn$ and $\SNPOn$ is contiguous (see Fig.~\ref{fig:opt_case2}~(a) and we can use similar argument as in Lemma~\ref{lem:sym}~(a) to deduce that
\[
W(\SNn \to S^{N+1}_n) = 2(\delta-(n-1+L)d).
\]
When $L \geq \lceil \frac{N+1}{2} \rceil$, we have
\[
F+(n-1)d-LNd < (n-1+L+LN)d.
\] 
This inequality is important because for $(n-1+L)d \leq \delta \leq F+(n-1)d-LNd$, the intersection of $\SNn$ and $\SNPOn$ is contiguous and the transition waste is
\[
W(\SNn \to S^{N+1}_n) = 2(\delta-(n-1+L)d),
\]
while for $F+(n-1)d-LNd < \delta < (n-1+L+LN)d$, the intersection between the two sets is non-contiguous (see Fig.~\ref{fig:opt_case2}~(b)) and the transition waste is
\[
W(\SNn \to S^{N+1}_n) = 2(N-L)F/N.
\]
Indeed, as the right endpoint of $\SNPOn$ is $(n-1)\frac{F}{N+1}+\frac{LF}{N+1}-2+\delta-F$ in this case, the intersection of the two sets has size
\[
\begin{split}
&\bigg(\bigg(\frac{(n-1)F}{N}+\frac{LF}{N}-1\bigg) - \bigg(\frac{(n-1)F}{N+1}+\delta\bigg) + 1\bigg)\\
&+\bigg(\bigg(\frac{(n-1)F}{N+1}+\frac{LF}{N+1}-1+\delta-F\bigg) -\frac{(n-1)F}{N} + 1\bigg)\\
&= \frac{LF}{N} + \frac{LF}{N+1} - F.
\end{split}
\]
Therefore, the transition waste is 
\[
\begin{split}
&W(\SNn\hspace{-3pt} \to\hspace{-2pt} S^{N+1}_n)\hspace{-2pt} =\hspace{-2pt} (|\SNn|\hspace{-2pt} +\hspace{-2pt} |\SNPOn|)\hspace{-2pt}-\hspace{-2pt}2|\SNn\hspace{-2pt}\cap\hspace{-2pt} \SNPOn|\hspace{-2pt} -\hspace{-2pt} \DNNPO\\
&= \bigg(\frac{LF}{N}\hspace{-2pt} +\hspace{-2pt} \frac{LF}{N+1}\bigg)-2\bigg(\frac{LF}{N}\hspace{-2pt} +\hspace{-2pt} \frac{LF}{N+1} - F\bigg)-\frac{LF}{N(N+1)}\\
&= 2(N-L)F/N. 
\end{split}
\]
These explain the formula of Sum~2.

\begin{figure}[!htb]
\centering
\includegraphics[scale=0.6]{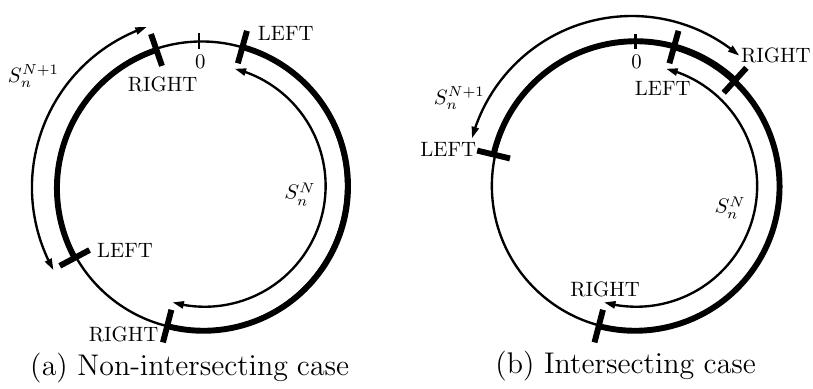}
\caption{Illustration of scenarios in Case 3.}
\label{fig:opt_case3}
\end{figure} 

\textbf{Case 3.} $(n-1+L+LN)d \leq \delta < F$. The right endpoint of $\SNn$ lies between its left endpoint and the left endpoint of $\SNPOn$. We divide this case further into two sub-cases, depending on whether the two sets intersect or not (see~Fig.~\ref{fig:opt_case3}). Note that when the two sets do intersect, the right endpoint of $\SNPOn$ is $(n-1)\frac{F}{N+1}+\frac{LF}{N+1}-2+\delta-F$ (i.e., having $-F$).

When $L < \lceil \frac{N+1}{2} \rceil$, for $(n-1+L+LN)d \leq \delta \leq F+(n-1)d-LNd$, the two sets do not intersect (see Fig.~\ref{fig:opt_case3}~(a)), and so, the transition waste is 
\[
W(\SNn \to S^{N+1}_n) = \frac{LF}{N}+\frac{LF}{N+1}-\frac{LF}{N(N+1)}= 2LNd,
\] 
while for $F+(n-1)d-LNd < \delta < F$, the two sets intersect~(see Fig.~\ref{fig:opt_case3}~(a)) and the transition waste is 
\[
\begin{split}
&W(\SNn \to\hspace{-2pt} S^{N+1}_n)\hspace{-2pt} =\hspace{-2pt} (|\SNn|\hspace{-2pt} +\hspace{-2pt} |\SNPOn|)\hspace{-2pt}-\hspace{-2pt}2|\SNn\hspace{-2pt}\cap\hspace{-2pt} \SNPOn|\hspace{-2pt} -\hspace{-2pt} \DNNPO\\
&= \bigg(\frac{LF}{N}\hspace{-2pt} +\hspace{-2pt} \frac{LF}{N+1}\bigg)-2\bigg(\hspace{-2pt}\bigg(\frac{(n-1)F}{N+1}\hspace{-2pt}+\hspace{-2pt}\frac{LF}{N+1}\hspace{-2pt}-\hspace{-2pt}1\hspace{-2pt}+\hspace{-2pt}\delta\hspace{-2pt}-\hspace{-2pt}F\bigg)\\ 
&\quad -\frac{(n-1)F}{N}+1\bigg)-\frac{LF}{N(N+1)} = 2(F+(n-1)d-\delta).
\end{split} 
\]

When $L \geq \lceil \frac{N+1}{2} \rceil$, the two sets $\SNn$ and $\SNPOn$ always intersect and the transition waste is $2(F+(n-1)d-\delta)$. 
These explain the formula of Sum3.
\end{proof} 

\begin{proof}[Proof of Theorem~\ref{thm:optimal}]
Lemma~\ref{lem:opt} establishes an \emph{implicit} formula for the transition waste when transitioning from a cyclic \NT~$\SCN$ to a $\delta$-shifted cyclic \NPOT~$\SDCNPO$. It remains to determine an \emph{explicit} form of the transition waste and show that it is minimized at $\delta_{\sf{opt}} = \left\lfloor \frac{N+L-1}{2}\right\rfloor d$. 
To simplify the computation, we assume that $\delta$ is divisible by $d \define \frac{F}{N(N+1)}$. 
Even with this simplification, the computation is still very tedious with many cases depending on the relation between $N$ and $L$ and the exact interval $\delta$ lies in (four cases, each has seven intervals to consider - Figs.~\ref{fig:case12},~\ref{fig:case34}). 

\begin{figure}[!htb]
\centering
\includegraphics[scale=0.8]{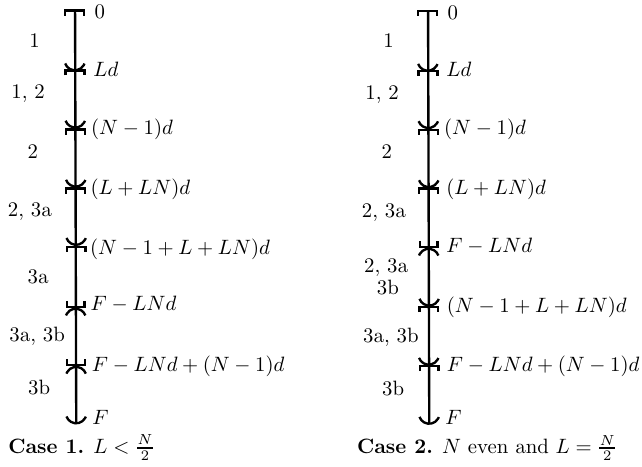}
\caption{Illustration of the intervals for $\delta$ and the non-empty sums contributing to the transition waste when $L < \lceil \frac{N+1}{2} \rceil$. 
The labels 3a/3b refer to the two component sums of \texttt{Sum3} (see Lemma~\ref{lem:opt}).
The appearance of the labels 1, 2, 3a, 3b in each interval indicate that these sums are non-empty in that interval of $\delta$.}
\label{fig:case12}
\end{figure}
Note that while the transition waste can be written as the sum of four component sums, depending on the interval that $\delta$ belongs to, only a few sums are \textit{non-empty} (the lower limit doesn't exceed the upper limit). We must know which sums are non-empty in which intervals of $\delta$ to obtain a precise formula for the transition waste. 
We provide below the explicit formulas of the transition wastes in all four cases and seven intervals, resulting in 28 sub-cases in total.
For each sub-case, given the expression of the transition waste, we identify the $\delta^*$ (divisible by $d$) in that interval that minimizes the transition waste and show that this minimum transition waste is greater than or equal to the transition waste provided in Theorem~\ref{thm:shifted_joining}.

\noindent\textbf{Case 1:} $L<\frac{N}{2}$ (see Fig.~\ref{fig:case12}).
\begin{itemize}
	\item \textbf{Case 1-a:} $0\leq \delta < Ld$. By Lemma~\ref{lem:opt}, and noting that only \texttt{Sum1} is non-empty, we have
	\[
	\begin{split}	
W(\SCN \to \SDCNPO) &= \texttt{Sum1} = 
 \sum_{n = \frac{\delta}{d}+2}^N 2((n-1)d-\delta)\\ &= \frac{\delta^2}{d}-(2N-1)\delta+N(N-1)d,
\end{split}
\]
which achieves its minimum value $(N-L+1)(N-L)d$ (among all $\delta$ divisible by $d$) at $\delta^*=(L-1)d$. 
This value is larger than the transition waste obtained in Theorem~\ref{thm:shifted_joining}, noting that $d = F/(N(N+1))$. 
	\item \textbf{Case 1-b:} $Ld \leq \delta < (N-1)d$. By Lemma~\ref{lem:opt}, 
\[
W(\SCN \to \SDCNPO) = \texttt{Sum1} + \texttt{Sum}2,
\]  
where 
\[
\begin{split}
\texttt{Sum1} &= \sum_{n = \delta/d+2}^N 2((n-1)d-\delta)\\ &= \frac{\delta^2}{d}-(2N-1)\delta+N(N-1)d,
\end{split}
\]
\[
\begin{split}
\texttt{Sum2} &= \sum_{n = 1}^{\delta/d-L+1} 2(\delta-(n-1+L)d)\\ &= \frac{\delta^2}{d}-(2L-1)\delta+L(L-1)d. 
\end{split}
\]
Note that it is important to determine the precise lower and upper limits for each sum. Hence, 
\begin{multline*}
W(\SNn \to S^{N+1}_n)\\ = \frac{2\delta^2}{d}-2(N+L-1)\delta+(N(N-1)+L(L-1))d.
\end{multline*}
This is a quadratic function of $\delta$, which achieves the minimum at $\delta_{\sf{opt}} = \left\lfloor \frac{N+L-1}{2}\right\rfloor d$. This is indeed the shift recommended in Theorem~\ref{thm:shifted_joining}. 
	\item \textbf{Case 1-c:} $(N-1)d \leq \delta < (L+LN)d$. We have
\begin{multline*}
W(\SCN \to \SDCNPO) = \texttt{Sum2}\\
= \begin{cases}
\sum_{n = 1}^{\frac{\delta}{d-L+1}} 2(\delta\hspace{-2pt}-\hspace{-2pt}(n\hspace{-2pt}-\hspace{-2pt}1\hspace{-2pt}+\hspace{-2pt}L)d),&\text{if } \delta\hspace{-2pt} \leq \hspace{-2pt}(N\hspace{-2pt}+\hspace{-2pt}L\hspace{-2pt}-\hspace{-2pt}1)d,\\
\sum_{n = 1}^{N} 2(\delta\hspace{-2pt}-\hspace{-2pt}(n\hspace{-2pt}-\hspace{-2pt}1\hspace{-2pt}+\hspace{-2pt}L)d)
,&\text{if } \delta\hspace{-2pt} \geq \hspace{-2pt}(N\hspace{-2pt}+\hspace{-2pt}L)d,
\end{cases}\\
= \begin{cases}
\frac{\delta^2}{d}\hspace{-2pt}-\hspace{-2pt}(2L\hspace{-2pt}-\hspace{-2pt}1)\delta\hspace{-2pt}+\hspace{-2pt}L(L\hspace{-2pt}-\hspace{-2pt}1)d,&\text{if } \delta \leq (N\hspace{-2pt}+\hspace{-2pt}L\hspace{-2pt}-\hspace{-2pt}1)d,\\
\underset{\delta=(N+L)d}{\geq} N(N\hspace{-2pt}+\hspace{-2pt}1)d = F,&\text{if } \delta\hspace{-2pt} \geq\hspace{-2pt} (N\hspace{-2pt}+\hspace{-2pt}L)d,
\end{cases}
\end{multline*}
which achieves its minimum value $\frac{(N-L)(N-L-1)F}{N(N+1)}$ (among all $\delta$ divisible by $d$) at $\delta^*=(N-1)d$. 
This value is larger than the transition waste obtained in Theorem~\ref{thm:shifted_joining}. 
	\item \textbf{Case 1-d:} $(L\hspace{-2pt}+\hspace{-2pt}LN)d\hspace{-2pt} \leq\hspace{-2pt} \delta\hspace{-2pt} <\hspace{-2pt} (N\hspace{-2pt}-\hspace{-2pt}1\hspace{-2pt}+\hspace{-2pt}L\hspace{-2pt}+\hspace{-2pt}LN)d$.
	We have
	\begin{multline*}
W(\SCN \to \SDCNPO) = \texttt{Sum2} + \texttt{Sum3a}\\
= \sum_{n = \delta/d-LN-L+2}^{N}\hspace{-15pt} 2(\delta-(n-1+L)d)
+ \sum_{n=1}^{\delta/d-LN-L+1}\hspace{-15pt}2LNd\\
= -\frac{\delta^2}{d} + (N-L-1+2LN)\delta\\ - (N^2L-N^2-NL-3N-2L+2)Ld,
\end{multline*}
which is minimized at either $\delta_1^*=(LN+L)d$ or $\delta_2^*=(N-2+L+LN)d$, i.e., $\min\{Ld(2N^2+3N-3),d(2N^2L+N+3L-2)\}$, which is greater than $2N^2Ld$, which in turn is larger than the transition waste in Theorem~\ref{thm:shifted_joining}.   
	\item \textbf{Case 1-e:} $(N\hspace{-2pt}-\hspace{-2pt}1\hspace{-2pt}+\hspace{-2pt}L\hspace{-2pt}+\hspace{-2pt}LN)d\hspace{-2pt} \leq\hspace{-2pt} \delta\hspace{-2pt} \leq\hspace{-2pt} F\hspace{-2pt}-\hspace{-2pt}LNd$.
	We have 
	\[
W(\SNn \to S^{N+1}_n) = \texttt{Sum3a}
= \sum_{n=1}^{N}2LNd = 2LN^2d,\\
\]
which is larger than the transition waste in Theorem~\ref{thm:shifted_joining}.   
	\item \textbf{Case 1-f:} $F\hspace{-2pt}-\hspace{-2pt}LNd\hspace{-2pt} < \hspace{-2pt}\delta\hspace{-2pt} \leq\hspace{-2pt} F\hspace{-2pt}-\hspace{-2pt}LNd\hspace{-2pt}+\hspace{-2pt}(N\hspace{-2pt}-\hspace{-2pt}1)d$.
	We have
	\begin{multline*}
W(\SCN \to \SDCNPO) = \texttt{Sum3a} + \texttt{Sum3b}\\
= \sum_{n = \delta/d-N^2-N+LN+1}^{N} \hspace{-20pt}2LNd
+ \sum_{n=1}^{\delta/d+LN-N^2-N}\hspace{-15pt}2(F+(n-1)d-\delta)\\
= -\frac{\delta^2}{d} + (2N^2-2LN+2N-1)\delta\\ - Nd(N^3-2N^2L+2N^2+NL^2-4NL+L-1),
\end{multline*}
which is minimized at either $\delta_1^*=F-LNd+d=(N^2+N-LN+1)d$ or $\delta_2^*=F-LNd+(N-1)d=(N^2+2N-LN-1)d$. Therefore, the minimum transition waste in this range of $\delta$ (assuming $\delta$ is divisible by $d$) is 
\begin{multline*}
\min\{2N^2Ld-2d,2N^2Ld-N^2d+Nd\}\\=2N^2Ld-N^2d+Nd,
\end{multline*}
which is greater than the transition waste in Theorem~\ref{thm:shifted_joining}.
	\item \textbf{Case 1-g:} $F-LNd+(N-1)d < \delta < F$. We have 
\[	
	\begin{split}
&W(\SCN \to \SDCNPO) = \texttt{Sum3b}
= \sum_{n=1}^{N} 2(F\hspace{-2pt}+\hspace{-2pt}(n\hspace{-2pt}-\hspace{-2pt}1)d\hspace{-2pt}-\hspace{-2pt}\delta)\\ &= -2N\delta\hspace{-2pt}+\hspace{-2pt}(2N^3d\hspace{-2pt}+\hspace{-2pt}3N^2d\hspace{-2pt}-\hspace{-2pt}Nd)\\
&\geq -2N(F\hspace{-2pt}-\hspace{-2pt}d)\hspace{-2pt}+\hspace{-2pt}(2N^3d\hspace{-2pt}+\hspace{-2pt}3N^2d\hspace{-2pt}-\hspace{-2pt}Nd) = N(N\hspace{-2pt}+\hspace{-2pt}1)d,
\end{split}
\]
which is greater than the transition waste in Theorem~\ref{thm:shifted_joining}.
\end{itemize} 

\noindent\textbf{Case 2:} $L = \frac{N}{2}$ and $N$ is even (see Fig.~\ref{fig:case12}).
\begin{itemize}
	\item \textbf{Case 2-a:} $0\leq \delta < Ld$. The formula of the transition waste is the same as Case 1-a.  
	\item \textbf{Case 2-b:} $Ld \leq \delta < (N-1)d$. The formula of the transition waste is the same as Case 1-b.  
	\item \textbf{Case 2-c:} $(N-1)d \leq \delta < (L+LN)d$. The formula of the transition waste is the same as Case 1-c.  
	\item \textbf{Case 2-d:} $(L+LN)d \leq \delta \leq F-LNd$.  The formula of the transition waste turns out to be the same as Case 1-d.  
	\item \textbf{Case 2-e:} $F\hspace{-2pt} -\hspace{-2pt}LNd\hspace{-2pt} < \hspace{-2pt}\delta \hspace{-2pt}< \hspace{-2pt}(N\hspace{-2pt}-\hspace{-2pt}1\hspace{-2pt}+\hspace{-2pt}L\hspace{-2pt}+\hspace{-2pt}LN)d$.
	We have 
	\[
	\begin{split}
&W(\SNn \to S^{N+1}_n) = \texttt{Sum2} + \texttt{Sum3a} + \texttt{Sum3b}\\ 
&\geq \texttt{Sum3a} = \sum_{n=\delta/d-N^2-N+LN+1}^{\delta/d-LN-L+1}2LNd\\ &= (N^2+N-2LN-L)2LNd \underset{L=N/2}{=} LN^2d,
\end{split}
\]
which is larger than the transition waste in Theorem~\ref{thm:shifted_joining}.   
	\item \textbf{Case 2-f:} \hspace{-2pt}$(N\hspace{-2pt}-\hspace{-2pt}1\hspace{-2pt}+\hspace{-2pt}L\hspace{-2pt}+\hspace{-2pt}LN)d\hspace{-2pt} \leq\hspace{-2pt} \delta\hspace{-2pt} <\hspace{-2pt} F\hspace{-2pt}-\hspace{-2pt}LNd\hspace{-2pt}+\hspace{-2pt}(N\hspace{-2pt}-\hspace{-2pt}1)d$.
	We have
	\[
	\begin{split}
&W(\SCN \to \SDCNPO) = \texttt{Sum3a} + \texttt{Sum3b} \geq \texttt{Sum3a}\\
&=\hspace{-10pt} \sum_{n = \delta/d-N^2-N+LN+1}^{N}\hspace{-30pt} 2LNd = 2(N^2\hspace{-2pt}+\hspace{-2pt}2N\hspace{-2pt}-\hspace{-2pt}LN\hspace{-2pt}-\hspace{-2pt}1\hspace{-2pt}-\hspace{-2pt}\delta/d)LNd\\
&\underset{\delta=F-LNd+(N-2)d}{\hspace{-65pt}\geq}\hspace{-65pt} 2\big(N^2\hspace{-2pt}+\hspace{-2pt}2N\hspace{-2pt}-\hspace{-2pt}LN\hspace{-2pt}-\hspace{-2pt}1\hspace{-2pt}-\hspace{-2pt}(N^2\hspace{-2pt}+\hspace{-2pt}N\hspace{-2pt}-\hspace{-2pt}LN\hspace{-2pt}+\hspace{-2pt}N\hspace{-2pt}-\hspace{-2pt}2)\big)LNd\\
&= 2LNd \underset{L=N/2}{=} N^2d,
\end{split}
\]
which is greater than the transition waste in Theorem~\ref{thm:shifted_joining}. 
	\item \textbf{Case 2-g:} $F-LNd+(N-1)d < \delta < F$. The formula of the transition waste is the same as Case 1-g.  
\end{itemize}

\noindent\textbf{Case 3:} $L> \frac{N+1}{2}$ (see Fig.~\ref{fig:case34}).
\begin{figure}[htb!]
\centering
\includegraphics[scale=0.8]{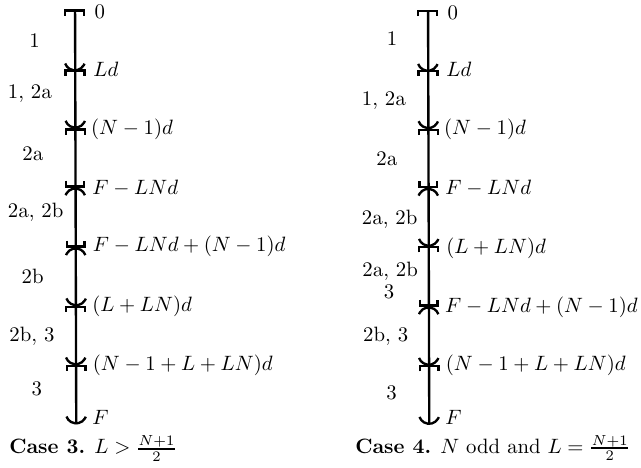}
\caption{Illustration of the intervals for $\delta$ and the non-empty sums contributing to the transition waste when $L \geq \lceil \frac{N+1}{2} \rceil$. The labels 2a/2b refer to the component sums of \texttt{Sum2} (see Lemma~\ref{lem:opt}).
The appearance of the labels 1, 2a, 2b, 3 in each interval indicate that these sums are non-empty in that interval. }
\label{fig:case34}\vspace{-10pt}
\end{figure} 

\begin{itemize}
	\item \textbf{Case 3-a:} $0\leq \delta < Ld$. The same as Case 1-a.  
	\item \textbf{Case 3-b:} $Ld \leq \delta < (N-1)d$. The same as Case 1-b. The minimum transition waste is achieved at $\delta^*=\lfloor \frac{N+L-1}{2}\rfloor d$, which is indeed the shift provided in Theorem~\ref{thm:shifted_joining}.
	\item \textbf{Case 3-c:} $(N\hspace{-2pt}-\hspace{-2pt}1)d\hspace{-2pt} \leq\hspace{-2pt} \delta\hspace{-2pt} \leq\hspace{-2pt} F\hspace{-2pt}-\hspace{-2pt}LNd$. The same as Case 1-c.
	\item \textbf{Case 3-d:} $F\hspace{-2pt}-\hspace{-2pt}LNd\hspace{-2pt} <\hspace{-2pt} \delta\hspace{-2pt} \leq\hspace{-2pt} F\hspace{-2pt}-\hspace{-2pt}LNd\hspace{-2pt}+\hspace{-2pt}(N\hspace{-2pt}-\hspace{-2pt}1)d$.
	We have \vspace{-5pt}
	\[
	\begin{split}
&W(\SCN \to \SDCNPO) = \texttt{Sum2a} + \texttt{Sum2b} \geq \texttt{Sum2a}\\
&= \sum_{n = 1}^{N} 2(\delta-(n-1+L)d)\\
&\underset{\delta=F-LNd+d}{\hspace{-40pt}\geq} 2N^2d(N-L)+(N^2d+Nd)\\ &\geq N(N+1)d = F,
\end{split}
\]
which is greater than the transition waste in Theorem~\ref{thm:shifted_joining}.   
	\item \textbf{Case 3-e:} $F-LNd+(N-1)d < \delta < (L+LN)d $.
	From Lemma~\ref{lem:opt}, we deduce that \vspace{-10pt}
	\[
W(\SNn \to S^{N+1}_n) = \texttt{Sum2b}
=\hspace{-10pt} \sum_{n=\delta/d-LN-L+2}^{\delta/d+LN-N^2-N}\hspace{-10pt} 2\frac{(N-L)F}{N},
\]
which is larger than the transition waste in Theorem~\ref{thm:shifted_joining} as long as there is at least one term in the sum. This can be easily shown by verifying that the upper limit is strictly larger than the lower limit of the sum.   
	\item \textbf{Case 3-f:} $(L\hspace{-2pt}+\hspace{-2pt}LN)d\hspace{-2pt} \leq\hspace{-2pt} \delta\hspace{-2pt} <\hspace{-2pt} (L\hspace{-2pt}+\hspace{-2pt}LN\hspace{-2pt}+\hspace{-2pt}N\hspace{-2pt}-\hspace{-2pt}1)d$. We have
	\[
W(\SCN\hspace{-3pt} \to\hspace{-1pt} \SDCNPO) = \texttt{Sum2b} + \texttt{Sum3} \geq \texttt{Sum2b},
\]
which is greater than the transition waste in Theorem~\ref{thm:shifted_joining}, using the same argument as Case 3-e.
	\item \textbf{Case 3-g:} $(L+LN+(N-1))d \leq \delta < F$. We have \vspace{-5pt} 
	\[
	\begin{split}
&W(\SCN\hspace{-3pt} \to\hspace{-1pt} \SDCNPO) = \texttt{Sum3}
= \sum_{n=1}^{N} 2(F+(n-1)d-\delta)\\ 
&\underset{\delta=F-d}{\geq} N(N+1)d = F,  \vspace{-5pt}
\end{split}
\]
which is greater than the transition waste in Theorem~\ref{thm:shifted_joining}.
\end{itemize} 

\noindent\textbf{Case 4:} $L=\frac{N+1}{2}$ and $N$ is odd (see Fig.~\ref{fig:case34}).
\begin{itemize}
	\item \textbf{Case 4-a:} $0\leq \delta < Ld$. The formula of the transition waste is the same as Case 3-a.  
	\item \textbf{Case 4-b:} $Ld \leq \delta < (N-1)d$. The formula of the transition waste is the same as Case 3-b.
	\item \textbf{Case 4-c:} $(N-1)d \leq \delta \leq F-LNd$. The formula of the transition waste is the same as Case 3-c.
	\item \textbf{Case 4-d:} $F-LNd < \delta < (L+LN)d$. This case can be settled using exactly the same argument as in Case 3-d.
	\item \textbf{Case 4-e:} $(L+LN)d \leq \delta \leq F-LNd+(N-1)d$. This case can be settled using exactly the same argument as in Case 3-e.
	\item \textbf{Case 4-f:} $F-LNd+(N-1)d < \delta < (N-1+L+LN)d$. This case can be settled using exactly the same argument as in Case 3-f.
	\item \textbf{Case 4-g:} The formula of the transition waste is the same as Case 3-g. \qedhere
\end{itemize} 
\end{proof}

\subsection{Frequently Used Notations}
\label{app:notations}

\begin{table}[htb!]
\centering
\tabcolsep=0.05cm
\begin{tabular}{|c|p{5.7cm}|}
    \hline
    \textbf{Notation} & \textbf{Meaning}\\
    \hline
    $N$ & The total number of machines in the system.\\ 
    \hline
    $n$ &
    The label of an individual machine. We have $n \in [N]\define \{1,2,\ldots,N\}$.\\
    \hline
    $L$ & The number of machines required by the underlying coded computing scheme. In the task allocation scheme, each task (index) is allocated to exactly $L$ different machines.\\
    \hline
    $F$ & The total number of tasks.\\
    \hline
    $f$ & The label of an individual task. We have $f\in [[F]] \define \{0,1,\ldots,F-1\}$.\\
    \hline
    $\SNn$ & A subset of $[[F]]$ representing the set of task indices allocated to Machine $n$ when the system has $N$ machines.\\
    \hline
    \NT & An ordered list of $N$ sets $\S^N=(S^N_1,\ldots,S^N_N)$ satisfying the $L$-Redundancy and the Load Balancing properties (see Definition~\ref{def:TAS}).\\ 
    \hline  
    $\SCN = (S^N_1,\ldots,S^N_N)$ & The order list of sets of task indices allocated to Machine $n\in [N]$ by a cyclic TAS (see~\eqref{eq:cyclic}).\\
    \hline
    $\SDCN = (S^N_1,\ldots,S^N_N)$ & The order list of sets of task indices allocated to Machine $n\in [N]$ by a $\delta$-shifted cyclic task allocation scheme (see Definition~\ref{def:shift}).\\
    \hline
    $\DNNP$ & $\DNNP\hspace{-2pt} \define\hspace{-2pt} |LF/N\hspace{-2pt}-\hspace{-2pt}LF/N'|$: the necessary load change when the system transitions from $N$ machines to $N'\hspace{-2pt}=\hspace{-2pt}N\hspace{-2pt}\pm\hspace{-2pt} 1$ ones. This is called \textit{necessary} as when a machine leaves/joins, even without any transition waste, the remaining ones must take more/less tasks to maintain the $L$-redundancy: every task must be covered by $L$ machines.\\
    \hline
    $W(\SNn \to \SNPn)$ & $W(\SNn \to \SNPn) \define |\SNn \Delta \SNPn|-\DNNP$: the transition waste incurred at Machine $n$ when transitioning from a set of tasks $\SNn$ to another set of tasks $\SNPn$ (see Definition~\ref{def:TW_one}).\\
    \hline
    $W_{n^*}(\SNn \to \SNMOn)$ & Same as $W(\SNn \to \SNPn)$ but more specific to the case $N'=N-1$ and Machine $n^*$ leaves.\\
    \hline
$W(\SN\to \SNPO)$ & $W(\SN\to \SNPO)\define \sum_{n\in [N]} W(\SNn\to \SNPOn)$ is the total transition waste at all machines when Machine $N+1$ joins.\\
	\hline 
	$W_{n^*}(\SN\to \SNMO)$ & $W_{n^*}(\SN\hspace{-3pt}\to\hspace{-3pt} \SNMO)\hspace{-3pt} \define \hspace{-3pt}\sum_{n \in [N]\setminus \{n^*\}}$ $W_{n^*}(\SNn\hspace{-3pt}\to\hspace{-3pt} \SNMOn)$: the total transition waste at all machines when Machine $n^*$ leaves.\\
	\hline
    $W_{\sf{avg}}(\SN\hspace{-4pt} \to\hspace{-2pt} \SNMO)$ & The average of $W_{n^*} (\SN\hspace{-5pt} \to\hspace{-3pt}  \SNMO)$ over $n^* \hspace{-3pt}\in\hspace{-3pt} [N]$.\\
	\hline
\end{tabular}
\caption{Frequently Used Notations.}
\label{tab:notations}
\end{table}


\newpage

\begin{IEEEbiographynophoto}
{Hoang Dau} (Member, IEEE) received the B.S. degree in applied
mathematics and informatics from Vietnam National
University, Hanoi, Vietnam, in 2006, and the M.S.
and Ph.D. degrees in mathematical sciences from
Nanyang Technological University, Singapore, in
2009 and 2012, respectively.
He is currently a senior lecturer in Computer Science at School of Computing Technologies, STEM College, RMIT University. His research interests include coding theory, discrete mathematics, and blockchain.
\end{IEEEbiographynophoto} 

\begin{IEEEbiographynophoto}
{Ryan Gabrys} (Member, IEEE) received the B.S. degree in mathematics and
computer science from the University of Illinois at Urbana-Champaing in
2005, and the Ph.D. degree in electrical engineering from the University of
California, Los Angeles in 2014. He is currently a Scientist jointly affiliated
with the Naval Information Warfare Center and the California Institute for
Telecommunications and Information Technology (Calit2) at the University
of California, San Diego. His research interests broadly lie in the areas
of theoretical computer science and electrical engineering, including coding
theory, combinatorics, and communication theory.
\end{IEEEbiographynophoto}

\begin{IEEEbiographynophoto}{Chen Feng} (Member, IEEE) received the B.Eng. degree from the Department of Electronic and Communications Engineering, Shanghai Jiao Tong University, China, in 2006, and the M.A.Sc. and Ph.D. degrees from the Department of Electrical and Computer Engineering, University of Toronto, Canada, in 2009 and 2014, respectively. From 2014 to 2015, he was a Post-Doctoral Fellow with Boston University, USA, and \'{E}cole Polytechnique F\'{e}d\'{e}rale de Lausanne (EPFL), Switzerland. He joined the School of Engineering, The University of British Columbia (UBC), Kelowna, Canada, in July 2015, where he is currently the Tier-2 Principal’s Research Chair in Blockchain and the Co-Cluster Lead of Blockchain at UBC. His research interests are in coding theory and its applications in various fields, including quantum communications and blockchain technology.
\end{IEEEbiographynophoto} 

\begin{IEEEbiographynophoto}{Yu-Chih Huang} (Member, IEEE) received the Ph.D. degree in electrical and computer engineering from Texas A\&M University (TAMU) in 2013. From 2013 to 2015, he was a Postdoctoral Research Associate with TAMU. In 2015, he joined the Department of Communication Engineering, National Taipei University, Taiwan, as an Assistant Professor and was promoted to an Associate Professor in 2018. In 2020, he joined the Institute of Communications Engineering, National Chiao Tung University (NCTU), Taiwan. He is currently an Associate Professor at National Yang Ming Chiao Tung University (the merger of National Yang Ming University and NCTU in 2021). His research interests are in information theory, coding theory, wireless communications, and statistical signal processing. He received the 2018 IEEE Information Theory Society Taipei Chapter and IEEE Communications Society Taipei/Tainan Chapter’s Best Paper Award for Young Scholars and was a recipient of the MOST Young Scholar Fellowship 2020. He is currently serving as an Associate Editor for IEEE Communications Letters.
\end{IEEEbiographynophoto}

\begin{IEEEbiographynophoto}{Quang-Hung Luu} (Member, IEEE) received the B.Sc. degree in applied mathematics and mechanics from Vietnam National
University, Hanoi, Vietnam, in 2004, and the Ph.D. degrees in earth and planetary sciences, and computer science and software engineering from Kyoto University and Swinburne University of Technology in 2012 and 2021, respectively. He is currently a research fellow sharing the time between Monash University and Swinburne University of Technology. His research interests include software testing, ocean modelling, connected and autonomous vehicles and data analysis.
\end{IEEEbiographynophoto}

\begin{IEEEbiographynophoto}{Eidah J. Alzahrani} is an assistant professor at Albaha University (Saudi Arabia). He obtained a bachelor's degree from Albaha University in 2007, a Master of Information Technology from La Trobe University (Australia) in 2010, and a PhD from RMIT University (Australia) in 2020, with a PhD thesis titled "Proactive auto-scaling techniques for containerised applications". His current research is on resource management for cloud computing data center, as well as Internet of Things (IoT) solutions in manufacturing.
\end{IEEEbiographynophoto}

\begin{IEEEbiographynophoto}{Zahir Tari} is a full professor at RMIT University (Australia) and the Research Director of
the RMIT Cyber Security Research and Innovation (CCSRI). His expertise is in the areas of system performance (e.g. P2P, Cloud,
Edge/IoT) and security (e.g.
SCADA, SmartGrid, Cloud, Edge/IoT).
\end{IEEEbiographynophoto}

\end{document}